%% file: 0main.tex
\documentclass[runningheads]{llncs}
\usepackage[T1]{fontenc}
\usepackage{graphicx}
\usepackage{amsmath,amsfonts,amssymb,bm,bbm}

\usepackage{tikz}
\usetikzlibrary{automata}
\usetikzlibrary{calc, automata, chains, arrows.meta}
\usetikzlibrary{chains,scopes}
\usetikzlibrary{positioning}
\usepackage{multirow,subfig}
\usepackage{bm}

\spnewtheorem{assumption}{Assumption}{\bfseries}{\itshape}

\newcommand{\bx}{\bm{x}}

\newcommand{\bz}{\bm{z}}
\newcommand{\bzero}{\bm{0}}

\DeclareMathOperator{\bone}{\mathbbm{1}}
\newcommand{\bR}{\mathbb{R}}
\newcommand{\bZ}{\mathbb{Z}}
\newcommand{\cA}{\mathcal{A}}
\newcommand{\cB}{\mathcal{B}}
\newcommand{\cD}{\mathcal{D}}

\newcommand{\cV}{\mathcal{V}}
\newcommand{\cW}{\mathcal{W}}

\newcommand{\hepsilon}{\widehat{\epsilon}}

\DeclareMathOperator*{\argmax}{arg\,max}

\DeclareMathOperator{\Exp}{\mathbb{E}}
\DeclareMathOperator{\Prob}{\mathbb{P}}

\begin{document}
\title{Best-Response Dynamics in Tullock Contests with Convex Costs}
%
%
\author{Abheek Ghosh\orcidID{0000-0002-4771-4612}}
%
%
\institute{Department of Computer Science, University of Oxford\\
  \email{abheek.ghosh@cs.ox.ac.uk}}
\maketitle              
\begin{abstract}
We study the convergence of best-response dynamics in Tullock contests with convex cost functions (these games always have a unique pure-strategy Nash equilibrium). We show that best-response dynamics rapidly converges to the equilibrium for homogeneous agents. For two homogeneous agents, we show convergence to an $\epsilon$-approximate equilibrium in $\Theta(\log\log(1/\epsilon))$ steps. For $n \ge 3$ agents, the dynamics is not unique because at each step $n-1 \ge 2$ agents can make non-trivial moves. We consider the model proposed by \cite{ghosh2023best}, where the agent making the move is randomly selected at each time step. We show convergence to an $\epsilon$-approximate equilibrium in $O(\beta \log(n/(\epsilon\delta)))$ steps with probability $1-\delta$, where $\beta$ is a parameter of the agent selection process, e.g., $\beta = n^2 \log(n)$ if agents are selected uniformly at random at each time step. We complement this result with a lower bound of $\Omega(n + \log(1/\epsilon)/\log(n))$ applicable for any agent selection process.

\end{abstract}
\input{1intro}
\input{2prelim}
\input{4hom-2}
\input{5dynamics}
\input{6hom-n}

\input{9conclusion}


%
%
%
\bibliographystyle{splncs04}
\bibliography{ref}

\appendix
\input{10appendix}

\end{document}

%% file: 1intro.tex
\section{Introduction}\label{sec:intro}
Contests are games where agents compete by making costly investments to win valuable prizes. Tullock~\cite{tullock1980efficient}'s model is one of the most widely used for studying these environments and has been applied to problems in economics, political science, and computer science~\cite{konrad2009strategy,vojnovic2015contest}. 
To give a concrete real-life application, consider the game among miners in proof-of-work cryptocurrencies like Bitcoin~\cite{chen2019axiomatic,leshno2020bitcoin}. A Bitcoin miner adds the next block (and collects the corresponding reward) with a probability proportional to her (costly) computational effort. A Tullock contest exactly models the game among these miners (assuming they are risk neutral): a set of agents compete to win a valuable prize by investing costly effort, and an agent gets a portion of the prize proportional to his effort.

We study the convergence of best-response (BR) dynamics in Tullock contests with convex costs. The existence of an equilibrium may not always be a good predictor of the agents' behavior. The traditional explanation of equilibrium is that it results from analysis and introspection by the agents in a situation where the rules of the game, the rationality of the agents, and the agents’ payoff functions are all common knowledge. Both conceptually and empirically, these theories may have problems~\cite{fudenberg1998theory}. A model of the outcome of a game is more plausible if it can be attained in a decentralized manner, ideally via a process that involves agents behaving in a natural-looking self-interested way. BR dynamics is arguably the simplest of these models. In BR dynamics, agents sequentially update their current strategy with one that best responds to that of the other agents. It is especially appropriate in settings, such as Tullock contests with convex costs, where pure-strategy equilibria are guaranteed to exist.

In our model for BR dynamics, we assume that at each time step one of the agents best responds and updates her current strategy. For two agents, this leads to a deterministic process where the agents play alternate moves.\footnote{Two consecutive moves by the same agent is redundant because the agent is already best responding after the first move and the second move does not change anything.} For more than two agents, the dynamics is not fully defined because all agents except the agent who made the most recent transition can make a non-trivial transition. So, we consider two models for selecting the agent that makes the BR move at each time step. First, we consider the randomized model introduced in \cite{ghosh2023best}. In this model, an agent $i$ makes the BR transition with probability $p_{t,i}(\bx_t)$ at time $t$, given the strategy profile at time $t$ is $\bx_t$.\footnote{We make the technical assumption that $p_{t,i}(x_t) > 0$ for all agents except the one who played the most recent move. See Section~\ref{sec:prelim:select} for more details.} In the second model, we assume that the agents are selected deterministically to ensure fast convergence to the equilibrium. We analyze the rate of convergence of the BR dynamics to an $\epsilon$-approximate equilibrium in these models.

\subsection{Our Results}\label{sec:results}
\sloppy
For two homogeneous agents with convex cost functions, we show that BR dynamics converges to an $\epsilon$-approximate pure-strategy Nash equilibrium in $\Theta(\log\log(1/\epsilon) + \log\log(1/\gamma))$ steps, where $\gamma$ is a natural function of the initial action profile (generally equal to the smallest positive action in the initial action profile). We also show faster convergence for restricted classes of convex cost functions.

For $n \ge 3$ homogeneous agents, for the randomized agent selection model, we show convergence in $O(\alpha \log(n/(\epsilon\delta)) + \beta \log\log(1/\gamma))$ steps with probability $1-\delta$, where $\alpha$ and $\beta$ are related to the randomized agent selection process (see Theorem~\ref{thm:hom}), e.g., $\alpha = n^2 \log(n)$ and $\beta = 1$ if agents are selected uniformly at random each time step.
We also provide a lower bound of $\Omega(\eta \log(n\delta) + \log(1/\epsilon)/\log(n) + \log\log(1/\gamma))$ with probability $1-\delta$, where $\eta$ is again related to the randomized agent selection process, e.g., $\eta = n$ for uniform selection. Our bounds are tight w.r.t. $\epsilon$ and $\gamma$, and are polynomially tight w.r.t. $n$. The $\Omega(\log(1/\epsilon))$ lower bound w.r.t. $\epsilon$ improves exponentially upon the previously known bound of $\Omega(\log\log(1/\epsilon))$~\cite{ghosh2023best}.
For the best-case agent selection model, we provide upper and lower bounds of $O(\log\log(1/\gamma) + n \log(n/\epsilon))$ and $\Omega(n + \log\log(1/\gamma) + \log(1/\epsilon)/\log(n))$, respectively.

Our rate of convergence analysis also implies convergence of the BR dynamics for any agent selection process that does not starve any agent, i.e., everyone gets to play infinitely often. A convergence result for the BR dynamics was not previously known for Tullock contests with homogeneous agents with general convex costs.
For non-homogeneous agents, we complement the non-convergence result of \cite{ghosh2023best} by providing an example with strictly convex cost functions and where the BR dynamics leads to a cycle that does not include a BR to $0$ (Appendix~\ref{sec:non-hom}). 

\paragraph{Discounted-sum dynamics.} 
One of the novel contributions of the paper is the introduction and analysis of the discounted-sum dynamics (Section~\ref{sec:dissum}), which is used in the analysis of the BR dynamics for $n \ge 3$ agents. This dynamics proceeds as follows: starting from an initial state $\bx_0 \in \bR^n$, at each time $t$, the dynamics picks a coordinate $i_t \in [n]$ and updates the value at the $i_t$-th coordinate with the negative discounted sum of the values at the remaining coordinates, i.e., $x_{t+1,i_t} = - \beta_t \sum_{j \neq i_t} x_{t,j}$. The discount factor $\beta_t \in [0,1)$ can be picked arbitrarily, possibly adversarially, given $t$, $i_t$, and $\bx_t$. We show that this dynamics rapidly converges to $0$ using a potential function. 


\paragraph{Intuition for the proof of convergence $n \ge 3$ agents.}
We divide the analysis into three main parts. 
The first part corresponds to a \textit{warm-up} phase of the dynamics. Intuitively, there is a certain domain $\cD$ such that the BR dynamics avoids corner cases, e.g., BR to $0$, and is easier to analyze if the output profile is in $\cD$. The warm-up phase corresponds to the time it takes to ensure that the output profile is inside $\cD$ (where it stays thereafter) and is proportional to the time it takes for every agent to make at least one move. 
In the second part of the analysis, we derive a set of sufficient conditions that ensures that the BR dynamics behaves like the discounted-sum dynamics mentioned earlier. We bound the time it takes for the BR dynamics to reach a state that satisfies these conditions, and show that the BR dynamics satisfies these conditions thereafter.
In the third part of the analysis, we use our results for the discounted-sum dynamics to show convergence of the output profile to the equilibrium output profile in the BR dynamics, which implies convergence to an approximate equilibrium by showing Lipschitzness of the utility function.

Some other interesting parts of our analysis include the lower bounds for $n \ge 3$ agents and the lower and upper bounds for $2$ agents that are tight up to additive constants. Omitted proofs are provided in the appendix.

\subsection{Related Work}
The papers by Ewerhart~\cite{ewerhart2017lottery} and Ghosh and Goldberg~\cite{ghosh2023best} are the closest to ours. These two papers study a special case of our model where the cost functions are linear, also known as lottery contests.

Ewerhart~\cite{ewerhart2017lottery} showed that a lottery contest with homogeneous agents is a BR potential game. 
The set of BR potential games~\cite{voorneveld2000best,kukushkin2004best,dubey2006strategic,uno2007nested,uno2011strategic,park2015potential} is a super-set of ordinal potential games, which itself is a super-set of exact potential games~\cite{monderer1996potential}.
In a BR potential game, given the profile $\bx_{-i}$ of all agents except $i$, the BR of agent $i$ matches the action that maximizes the potential function given $\bx_{-i}$. 
Notice that being a BR potential game does not imply anything about actions that are not BRs, i.e., actions that increase an agent's utility but are not BRs may decrease the potential function and vice-versa. Lottery contests do not have an exact potential, so they are less structured than exact potential games and are strictly outside the class~\cite{ewerhart2017lottery}. Lottery contests with non-homogeneous agents have also been proven not to be ordinal potential games~\cite{ewerhart2020ordinal}.

Ghosh and Goldberg~\cite{ghosh2023best} used the BR potential function of Ewerhart~\cite{ewerhart2017lottery} and derived bounds on the rate of convergence of BR dynamics in lottery contests with homogeneous agents. They showed that the potential function of Ewerhart~\cite{ewerhart2017lottery} is strongly convex and smooth in a region around the equilibrium point. They also show that the BR dynamics frequently visits this region. Outside this region, the BR potential function never increases, and inside this region, using techniques related to coordinate descent~\cite{wright2015coordinate}, they show rapid improvement in the potential. Our analysis is quite different from that of Ghosh and Goldberg~\cite{ghosh2023best}. E.g., we do not use a BR potential function or the ideas from convex optimization (coordinate descent) in our analysis.\footnote{The potential function used to analyze the discounted-sum dynamics may be considered a \textit{partial} potential function for the BR dynamics; partial because (i) it holds only when the output profile is in a certain restricted region, and (ii) in this restricted region, a BR move does not increase the potential but may also not decrease it (in this paper, the potential is supposed to decrease, unlike standard notation where potential increases).} In fact, for Tullock contests with general convex cost functions, there does not exist a closed-form formula for the BR of an agent. Rather, our analysis uses the discounted-sum dynamics as discussed earlier. We do partially borrow the broad framework that Ghosh and Goldberg~\cite{ghosh2023best} use to handle corner cases like BR to $0$, but this also needs substantial generalization to handle arbitrary convex cost functions.

\cite{moulin1978strategically} implicitly shows strategic equivalence between contests and zero-sum games. This directly implies convergence of fictitious play dynamics for two agents~\cite{ewerhart2020fictitious}, but no such result has been proven for three or more agents. A Tullock contest corresponds to a Cournot game with isoelastic inverse demand and constant marginal costs. There are convergence results of learning dynamics for specific types of Cournot games, like Cournot oligopoly with strictly declining BR functions~\cite{deschamps1975algorithm,thorlund1990iterative}, Cournot game with linear demand~\cite{slade1994does}, aggregative games that allow monotone BR selections~\cite{huang2002fictitious,dubey2006strategic,jensen2010aggregative}, and others~\cite{dragone2012static,bourles2017altruism}. However, all these methods do not apply to the Tullock contest whose BR function is not monotone~\cite{dixit1987strategic}. 
A different line of research has studied the convergence (or chaotic behavior) of learning dynamics in other types of contests (like all-pay auctions) and Cournot games (e.g.,~\cite{puu1991chaos,warneryd2018chaotic,cheung2021learning}), but these techniques and results also do not apply to Tullock contests.

%% file: 2prelim.tex
\section{Preliminaries}\label{sec:prelim} 
Let $[n] = \{1, 2, \ldots, n\}$. Let $\lg(x) = \log_2(x)$ and $\ln(x) = \log_e(x)$.

\subsection{Tullock Contest}
In the model by Tullock~\cite{tullock1980efficient}, there are $n$ agents and one contest with a unit prize (normalized). The agents simultaneously produce non-negative output $\bx = (x_1, \ldots, x_n) \in \bR_{\ge 0}^n$. Let $\bx_{-i} = (x_1, \ldots, x_{i-1}, x_{i+1}, \ldots, x_n)$, $s = \sum_{j \in [n]} x_j$, and $s_{-i} = \sum_{j \neq i} x_j$.
Agent $i$ incurs a cost of $c_i(x_i)$ for producing the output $x_i$. Agent $i$ gets an amount of prize proportional to $x_i$ if at least one agent produces a strictly positive output, else it is $1/n$.\footnote{Some papers in the literature, e.g., \cite{dasgupta1998designing}, assume that all agents get a prize of $0$ if they all produce a $0$ output. Our analysis and results remain the same with this alternate assumption as well.} The utility of agent $i$, $u_i(\bx)$, is
\begin{equation}\label{eq:utility:single}
    u_i(\bx) = \frac{x_i}{\sum_j x_j} - c_i(x_i) = \frac{x_i}{x_i + s_{-i}} - c_i(x_i).
\end{equation}
Notice that the utility of agent $i$ depends upon her output $x_i$ and the total output of other agents $s_{-i}$ but not upon the distribution of $s_{-i}$ across the $(n-1)$ agents. To signify this, we shall denote $u_i(\bx)$ as $u_i(x_i, s_{-i})$.

In this paper, we restrict our attention to convex cost functions. 
\begin{assumption}\label{as:cost1}
    For every agent $i \in [n]$, we assume that her cost function $c_i$ satisfies the following conditions 
    \begin{itemize}
        \item twice differentiable;
        \item increasing, $c_i'(z) > 0$ for all $z \ge 0$;
        \item weakly convex, $c_i''(z) \ge 0$ for all $z \ge 0$;
        \item Lipschitz continuous, $|c_i(z) - c_i(\overline{z})| \le K |z - \overline{z}|$ for all $z, \overline{z}$ in the neighborhood of the equilibrium,\footnote{For convergence to an $\epsilon$-approximate equilibrium, we will assume Lipschitz continuity in a neighborhood of size proportional to $\epsilon$.} where $K$ is the Lipschitz constant;
        \item zero cost for not participating, $c_i(0) = 0$, normalized.
    \end{itemize}
\end{assumption}
The conditions in Assumption~\ref{as:cost1}, except the Lipschitz continuity, are standard assumptions in the literature. Without these assumptions, a pure-strategy Nash equilibrium may not exist (see, e.g., Chapter~4 of \cite{vojnovic2015contest}, \cite{ewerhart2020unique}, and references therein). Lipschitz continuity is used only in Lemma~\ref{lm:lipschitz} to ensure that if an output profile is close to the equilibrium output profile, then the output profile is an approximate equilibrium.

\begin{remark}\label{rk:logit}
Our model captures contests of general logit form with concave success functions and convex costs~\cite{vojnovic2015contest}. In logit contests, the utility of agent $i$ is given by $\widehat{u}_i(\widehat{\bx}) = \frac{\widehat{f_i}(\widehat{x_i})}{\sum_j \widehat{f_j}(\widehat{x_j})} - \widehat{c_i}(\widehat{x_i})$, where $\widehat{x_j}$ is the output of agent $j$. For each $i$, if $\widehat{f_i}$ is concave and $\widehat{c_i}$ is convex (and both are non-negative and strictly increasing), then we can do the change of variable $x_i = \widehat{f_i}(\widehat{x_i})$ and set $c_i(x_i) = \widehat{c_i}(\widehat{f_i}^{-1}(x_i))$ to write the utility function as $u_i(\bx)  = \frac{x_i}{\sum_j x_j} - c_i(x_i) = \widehat{u_i}(\widehat{\bx})$. 
Similarly, a utility function of the form $\frac{V_i x_i}{\sum_j x_j} - c_i(x_i)$, where $V_i$ is the value of the prize for agent $i$, can be converted to the form given in \eqref{eq:utility:single} by scaling down by $V_i$, which does not affect the strategies of the agents.
Note that general Tullock contests with utility functions of the form $\widehat{u}_i(\widehat{\bx}) = \frac{\widehat{x_i}^g}{\sum_j \widehat{x_j}^{g}} - \widehat{c_i}(\widehat{x_i})$ for some $g \in (0,1]$ are special cases of logit contests with $\widehat{f_i}(\widehat{x_i}) = \widehat{x_i}^g$.
\end{remark}

\subsection{Best-Response Dynamics}\label{sec:prelim:br}
Given the total output $s_{-i} = \sum_{j \neq i} x_j$ of all agents except $i$, the best response (BR) of agent $i$ is an action $x_i$ such that
\begin{equation*}
    x_i \in \argmax_{z \ge 0} u_i(z, s_{-i}) = \argmax_{z \ge 0} \left( \frac{z}{z + s_{-i}} - c_i(z) \right) .
\end{equation*}
First, notice that an agent has no BR if the output produced by every other agent is $0$. If $s_{-i} = 0$, then by producing an output of $\epsilon > 0$, agent $i$ gets a utility of $u_i(\epsilon, 0) = 1 - c_i(\epsilon)$, which is strictly more than $u_i(0, 0) = 1/n$ for small $\epsilon$ (because $c_i$ is continuous and $c_i(0) = 0$, which implies $c_i(\epsilon)$ is close to $0$), so $x_i = 0$ cannot be a BR. Further, any $\epsilon > 0$ cannot be a BR because $u_i(\epsilon/2, 0) = 1 - c_i(\epsilon/2) > 1 - c_i(\epsilon) = u_i(\epsilon, 0)$ (and because $c_i$ is increasing).
To circumvent this technical issue, we assume that there is a small positive value $a < 1$ such that $x_i = a$ is the BR of any agent $i$ if $s_{-i} = 0$. We shall work with this assumption---an agent plays $a$ if everyone else plays $0$---throughout our paper.\footnote{This issue of not having any BR to $0$ can also be resolved by an alternate technical assumption: the prize is given to agent $i$ with probability $\frac{x_i}{b + \sum_j x_j}$ and agent $i$'s expected utility is $\frac{x_i}{b + \sum_j x_j} - c_i(x_i)$, where $b$ is a small positive constant. Notice that, with this alternate assumption, the prize may not get allocated to any agent (or is allocated to a pseudo agent) with a positive probability of $\frac{b}{b + \sum_j x_j}$, unlike our model. We \textit{expect} all results in this paper to hold for this alternate model as well.\label{fn:alternateModel}}

On the other hand, agent $i$ has a unique BR if $s_{-i} > 0$, i.e., if the output produced by at least one other agent $j$ is non-zero. This unique BR can be computed by taking a derivative of $u_i(z, s_{-i})$ with respect to $z$. If $s_{-i} > 0$, then 
\begin{align}
    \frac{\partial u_i(z, s_{-i})}{\partial z} &= \frac{ s_{-i} }{(z + s_{-i})^2} - c_i'(z), \label{eq:du} \\
    \frac{\partial^2 u_i(z, s_{-i})}{\partial z^2} &= \frac{-2 s_{-i} }{(z + s_{-i})^3} - c_i''(z) < 0, \label{eq:ddu}
\end{align}
where the last inequality is true because $z \ge 0$, $c_i''(z) \ge 0$, and $s_{-i} > 0$. So, $u_i(z, s_{-i})$ is strictly concave w.r.t. $z$, i.e., $\frac{\partial u_i(z, s_{-i})}{\partial z}$ is strictly decreasing w.r.t. $z$. 
Let $BR_i(s_{-i})$ denote the BR of agent $i$ given the total output of other agents is $s_{-i}$. The first-order conditions for the BR are
\begin{align}\label{eq:br}
    BR_i(s_{-i}) > 0 \text{ and } \left. \frac{\partial u_i(z, s_{-i})}{\partial z} \right|_{z = BR_i(s_{-i})} = 0, \text{ if } \left. \frac{\partial u_i(z, s_{-i})}{\partial z} \right|_{z = 0} > 0; \\
    BR_i(s_{-i}) = 0 \text{ and } \left. \frac{\partial u_i(z, s_{-i})}{\partial z} \right|_{z = BR_i(s_{-i})} \le 0, \text{ if } \left. \frac{\partial u_i(z, s_{-i})}{\partial z} \right|_{z = 0} \le 0.
\end{align}
If $s_{-i} > 0$, notice that at $z = \max(1, (c_i')^{-1}(1))$, the agent has a cost of at least $1$ but a prize of less than $1$, therefore $ \left. \frac{\partial u_i(z, s_{-i})}{\partial z} \right|_{z = \max(1, c_i^{-1}(1))} < 0$. So, the BR is less than $\max(1, c_i^{-1}(1))$. For $s_{-i} = 0$, the BR is $a < 1$ by assumption.

Slightly overloading notation, let $\bx_t = (x_{t,i})_{i \in [n]}$ denote the action profile of the agents at time $t$ in the BR dynamics. Similarly, let $s_t = \sum_j x_{t,j}$ and $s_{t,-i} = \sum_{j \neq i} x_{t,j}$. The BR dynamics starts with an initial profile $\bx_0 = (x_{0,i})_{i \in [n]}$. At each time step $t \ge 0$, an agent $i_t \in [n]$ makes a BR move. 
Formally, $x_{t+1,i_t} = BR_{i_t}(s_{t, -i_t})$ and $x_{t+1, j} = x_{t, j}$ for $j \neq i_t$. In this paper, we study the convergence (or non-convergence) of this BR dynamics.

\subsection{Agent-Selection Models for $n \ge 3$ Agents}\label{sec:prelim:select}
When there are just two agents, the BR dynamics proceeds uniquely: the two agents make alternate BR moves because making consecutive moves is redundant. But, when there are $n \ge 3$ agents, the BR dynamics is not unique. At any time step, the $n-1 \ge 2$ agents who did not make the most recent BR move can possibly make a non-trivial move. The rate of convergence depends upon which agent makes a move at a given time step. We consider two models for selecting the agent playing the BR move at each time step.

\subsubsection{Randomized Agent-Selection Model.}
In the random selection model~\cite{ghosh2023best}, agent $i$ makes the BR move w.p. (with probability) $p_{t,i}(\bx_t)$ at time $t$ given output profile $\bx_t$. The probability $p_{t,i}(\bx_t)$ models the relative activity of agent $i$ at time $t$ given the current profile $\bx_t$. 
We assume that $p_{t,i}(\bx_t) \ge L > 0$ for all agents except the agent who made the last transition. 
Our rate of convergence analysis will be worst-case over all $p_{t,i}$ functions given the parameter $L > 0$.
Even if this requirement on the $p_{t,i}$ functions is not satisfied, our analysis still implies convergence of the BR dynamics assuming none of the agents are starved (every agent gets to play infinitely often).

An important special case of the model is to assume that $p_{t,i}(\bx_t) = 1/n$ for all $t$, $i$, and $\bx_t$, i.e., the agent making the move is selected uniformly at random. 

\subsubsection{Best-Case Agent-Selection Model.}
In the best-case selection model, we assume that the selection process picks the sequence of agents making the BR moves to reach the $\epsilon$-approximate equilibrium as quickly as possible. We provide lower bounds on the time till convergence for this model, which by definition automatically apply to all other models, including the randomized model. We also prove almost tight upper bounds for this model. 
On the other hand, our upper bounds for the randomized model are tight w.r.t. $\epsilon$ (and the initial state) but only polynomially tight w.r.t. $n$ (e.g., the lower bound is $\Omega(n)$ while the upper bound of $O((n \log(n))^2)$ for uniform randomized selection).

\subsection{Equilibrium}
A Tullock contest with concave utility always has a pure-strategy Nash equilibrium (which is also the unique equilibrium, including mixed-strategy Nash equilibria~\cite{vojnovic2015contest,ewerhart2020unique}). So, we exclusively focus on pure equilibria in this paper.
\begin{definition}[Pure-Strategy Nash Equilibrium]
    An action profile $\bx = (x_1, \ldots, x_n)$ is a pure-strategy Nash equilibrium if it satisfies
    \[
        u_i (x_i, \bx_{-i}) \ge u_i (x_i', \bx_{-i}), 
    \]
    for every agent $i$ and every action $x_i'$ for agent $i$.
\end{definition}
In general, BR dynamics in a Tullock contest never exactly reaches the equilibrium, rather it may \textit{converge} to the equilibrium. The dynamics converges to an equilibrium if it reaches an $\epsilon$-approximate equilibrium in finite time for any $\epsilon > 0$.
\begin{definition}[Approximate Pure-Strategy Nash Equilibrium]
    An action profile $\bx = (x_1, \ldots, x_n)$ is an $\epsilon$-approximate pure-strategy Nash equilibrium, for $\epsilon > 0$, if it satisfies
    \[
        u_i (x_i, \bx_{-i}) \ge (1- \epsilon) u_i (x_i', \bx_{-i}), 
    \]
    for every agent $i$ and every action $x_i'$ for agent $i$.
\end{definition}

\subsection{Homogeneous Agents}
The agents are homogeneous if they all have the same cost function. Additionally, to make our analysis cleaner, we make the following assumption on the (homogeneous) cost function
\begin{assumption}\label{as:cost2}
    For every agent $i$, $c_i(x_i) = \frac{n-1}{n^2} c(x_i)$, where $c$ satisfies $c'(1) = 1$ (and Assumption~\ref{as:cost1}). 
\end{assumption}
Assumption~\ref{as:cost2} is w.l.o.g. by suitable change of notation.\footnote{Say the (homogeneous) cost of agent $i$ is $c_i(x_i) = \tilde{c}(x_i)$ for some $\tilde{c}$, for every $i$. First, we can introduce the $\frac{n-1}{n^2}$ factor by writing $c_i(x_i) = \frac{n-1}{n^2} \widehat{c}(x_i)$, where $\widehat{c}(x_i) = \frac{n^2}{n-1} \tilde{c}(x_i)$. Then, let $\gamma > 0$ be the solution to $\widehat{c}'(\gamma) = \frac{1}{\gamma}$. Such a $\gamma$ always exists because $\widehat{c}'$ is increasing, while $\frac{1}{\gamma}$ is decreasing and converges to $\infty$ as $\gamma$ converges to $0$. So, setting $c(x_i) = \widehat{c}(\gamma x_i)$, we get $c'(1) = \gamma \widehat{c}'(\gamma) = 1$. Finally, we can change the notation and say that an agent $i$ produces an output of $y_i = \gamma x_i$ instead of $x_i$.}
The utility of agent $i$ can be written as $u_i(x_i, s_{-i}) = \frac{x_i}{x_i + s_{-i}} - \frac{n-1}{n^2} c(x_i)$ for every $i \in [n]$. Notice that the BR of each agent is now identical given the same action profile of the remaining agents; let us, therefore, denote the best response function as $BR = BR_i$ for every $i \in [n]$ by suppressing $i$.
There is a unique equilibrium where each agent $i$ plays the same action (e.g., see \cite{vojnovic2015contest}, Chapter~4), and given Assumption~\ref{as:cost2}, this is equal to $1$. Because, if $x_j = 1$ for all $j$, then $s_{-i} = n-1$ and $\left. \frac{\partial u_i(z, s_{-i})}{\partial z} \right|_{z = 1} = \frac{s_{-i} }{(1 + s_{-i})^2} - \frac{n-1}{n^2} c'(1) = \frac{ n-1 }{(1 + n-1)^2} - \frac{n-1}{n^2} = 0$.

%% file: 4hom-2.tex
\section{Two Homogeneous Agents}\label{sec:hom_2_agents}
In this section, we study the rate of convergence of best-response dynamics for two homogeneous agents.
As introduced earlier, the initial state is denoted by $\bx_0 = (x_{0,1}, x_{0,2})$, and the state after $t \ge 0$ best-response moves is denoted by $\bx_t = (x_{t,1}, x_{t,2})$. 
We assume that agent $1$ makes the best-response move when $t$ is odd, i.e., $i_t = 1$ if $t$ odd, and agent $2$ makes the best-response move when $t$ is even, i.e., $i_t = 2$ if $t$ even. This is w.l.o.g. because: (i) it is redundant for an agent to play two consecutive best-response moves, and therefore the two agents should alternate; (ii) as the agents are homogeneous, it does not matter who makes the first move. Formally, for $t = 0, 1, 2, \ldots$, the BR dynamics proceeds as:
\begin{align*}
    x_{t+1, 2} &= BR(x_{t, 1}), \text{ and } x_{t+1, 1} = x_{t, 1},\quad \text{ if $t = 0, 2, 4, \ldots$,} \\
    x_{t+1, 1} &= BR(x_{t, 2}), \text{ and } x_{t+1, 2} = x_{t, 2},\quad \text{ if $t = 1, 3, 5, \ldots$.}
\end{align*}
Next, we prove an $O(\log\log)$ rate of convergence for two agents.

\begin{theorem}\label{thm:hom2}
Best-response dynamics in Tullock contests with two homogeneous agents reaches an $\epsilon$-approximate equilibrium within $\lg\lg(\frac{1}{\epsilon}) + \lg\lg(\frac{1}{\gamma}) + O(1)$ steps, where $\gamma$ is a function of the initial state: $\gamma = x_{0,1}$ if $0 < x_{0,1} < 1$; $\gamma = (c')^{-1} \left( \frac{1}{x_{0,1}} \right)$ if $x_{0,1} > 1$ and $c'(0) < \frac{4}{x_{0,1}}$; and $\gamma = a$ otherwise.
\end{theorem}
\input{proofs/2hom.tex}

\cite{ghosh2023best} provided a lower bound of $\lg\lg(\frac{1}{\epsilon}) + \lg\lg(\frac{1}{\gamma}) + \Omega(1)$ for linear cost functions. So, for the class of all convex cost functions, our bound is tight up to additive constants. But, faster convergence rates are possible for specific classes of convex cost functions. For example, if $c'(z) = z^q$ and $q \rightarrow \infty$, then $BR(x_{-i}) \rightarrow 1$ for any $x_{-i} \in [0,1]$ and $i \in [2]$. So, within a constant number of steps, we can approximately reach the equilibrium.
In Theorem~\ref{thm:hom2:curved}, we provide tight bounds for cost functions whose derivatives can be bounded as $z^p \le c'(z) \le z^q$ for $z \in [0,1]$ and $p \ge q \ge 0$. Notice that the case $p = q = 1/r$ for $r \in (0,1]$ is equivalent to Tullock contests with concave contest success functions ($\frac{x_i^r}{\sum_j x_j^r}$) and linear costs.
\begin{theorem}\label{thm:hom2:curved}
Best-response dynamics in Tullock contests with two homogeneous agents and a convex cost function $c$ that satisfies $z^p \le c'(z) \le z^q$ for $z \in [0,1]$ and $p \ge q \ge 0$ reaches an $\epsilon$-approximate equilibrium in $\lg\lg(\frac{1}{\epsilon}) + \frac{1}{\lg(2+r)}\lg\lg(\frac{1}{\gamma}) - \lg\lg(2+r) + \Theta(1)$ steps, where $r \in [q,p]$ and $\gamma$ is a function of the initial state: $\gamma = x_{0,1}$ if $0 < x_{0,1} < 1$; $\gamma = (\frac{1}{x_{0,1}})^{1/r}$ if $x_{0,1} > 1$ and $c'(0) < \frac{4}{x_{0,1}}$; and $\gamma = a$ otherwise.
\end{theorem}

%% file: proofs/2hom.tex
\begin{proof}
Notice that the evolution of the states of the BR dynamics $(\bx_t)_{t \ge 0}$ can be tracked by a single sequence $(z_t)_{t \ge 0}$ defined as follows: 
\[
    z_t =
    \begin{cases}
        x_{t,1} &\text{ if $t = 0, 2, 4, \ldots$,} \\
        x_{t,2} &\text{ if $t = 1, 3, 5, \ldots$.}
    \end{cases}
\]
$z_t$ has a one-to-one correspondence with $\bx_t$ because: (i) Only the action of agent $1$ in the initial profile is relevant for the BR dynamics, because, at $t=0$, agent $2$ makes the BR move, which is a function of only $x_{0,1}$, and which overwrites $x_{0,2}$. (ii) For each time step, one of the agents moves and the other stays at their current action, and $z_t$ is equal to the new action at each time point.

We next prove the following properties about the sequence $(z_t)_{t \ge 0}$: (i) $z_t \le 1$ for all $t \ge 1$; (ii) if $z_t < 1$, then $z_t < z_{t+1} < 1$. 
\begin{lemma}\label{lm:<=1}
    The BR of agent $i$, $BR(x_{-i})$, given any strategy $x_{-i}$ of the other agent, satisfies the following properties:
    \begin{enumerate}
        \item \label{lm:<=1:1} if $x_{-i} < 1$, then $x_{-i} < BR(x_{-i}) < 1$,
        \item \label{lm:<=1:2} if $x_{-i} = 1$, then $BR(x_{-i}) = 1$, and
        \item \label{lm:<=1:3} if $x_{-i} > 1$, then $BR(x_{-i}) < 1$.
    \end{enumerate}
\end{lemma}
\paragraph{Dependency on the initial state.} Let us now analyze the effect of the initial state $z_0$ on the BR dynamics. If $z_0 = 1$, from Lemma~\ref{lm:<=1} \eqref{lm:<=1:2}, we know that $z_t = 1$ for all $t$. This can alternately be deduced from the fact that $(1, 1)$ is the equilibrium profile. For the remaining proof, we assume that $z_0 \neq 1$. As $z_0 \neq 1$, from Lemma~\ref{lm:<=1} \eqref{lm:<=1:1} and \eqref{lm:<=1:3}, we know that $z_1 < 1$. Again using Lemma~\ref{lm:<=1} \eqref{lm:<=1:1}, we know that $z_{t} < z_{t+1} < 1$ for all $t \ge 1$, i.e., $(z_t)_{t \ge 1}$ is a strictly increasing sequence going towards $1$. Let us define the variable $\gamma$ used in the theorem statement as follows:
\begin{itemize}
    \item If $z_0 = 0$, set $\gamma = a < 1$. Notice that $z_1 = a = \gamma$.
    \item If $0 < z_0 < 1$, set $\gamma = z_0 < 1$.
    \item If $z_0 > 1$ and $c'(0) < \frac{4}{z_0}$ set $\gamma = (c')^{-1} \left( \frac{1}{z_0} \right)$. From the first-order condition, equation~\eqref{eq:br}, and using $0 \le z_1 < 1 < z_0$, we have
    \begin{multline}\label{eq:thm:hom2:1}
        \frac{z_0}{(z_0 + z_1)^2} = \frac{1}{4}c'(z_1) \implies \frac{z_0}{(z_0 + z_0)^2} = \frac{1}{4 z_0} \le \frac{1}{4}c'(z_1) \le \frac{1}{z_0} = \frac{z_0}{(z_0 + 0)^2} \\
        \implies z_1 = (c')^{-1} \left( \frac{\kappa}{z_0} \right) \text{ for some $\kappa \in [1,4]$.}
    \end{multline}
    We will later see that the dependency on $\gamma$ is $\lg\lg(\frac{1}{\gamma}) + O(1)$, so we may ignore the constant $\kappa$ as it can be consumed in the $O(1)$ term.
    \item If $z_0 > 1$ and $c'(0) \ge \frac{4}{z_0}$, set $\gamma = a$. In this case, it is easy to check that $z_1 = 0$ and $z_2 = \gamma$. 
\end{itemize}

By the definition of $\gamma$, either $z_0$, $z_1$, or $z_2$ is equal to $\gamma < 1$.
So, for the remaining portion of the proof, by shifting the time by at most two steps, w.l.o.g., let us assume that $z_0 = \gamma < 1$.
We know that the sequence $(z_t)_{t \ge 0}$ is strictly increasing. We next try to find the number of steps required by the sequence to increase from $\gamma$ to $1 - \epsilon$. We break this into two parts: the number of steps required to reach (i) $\frac{1}{2}$ from $\gamma$ in Lemma~\ref{lm:2hom:1}; (ii) $1 - \epsilon$ from $\frac{1}{2}$ in Lemma~\ref{lm:2hom:2}.
\begin{lemma}\label{lm:2hom:1}
    Given $z_0 = \gamma \in (0,\frac{1}{2}]$, $z_t \ge \frac{1}{2}$ for all $t \ge \lg\lg(\frac{1}{\gamma})$.
\end{lemma}
\begin{lemma}\label{lm:2hom:2}
    Given $z_0 \in [\frac{1}{2}, 1]$, $z_t \ge 1 - \epsilon$ for all $t \ge \lg\lg(\frac{1}{\epsilon}) + O(1)$.
\end{lemma}

Combining the previous two lemmas, we know that after $t-1$ steps, where $t-1 \ge \lg\lg(\frac{1}{\epsilon}) + \lg\lg(\frac{1}{\gamma}) + O(1)$, we get $z_{t-1} \ge 1-\epsilon$ and $z_t \ge 1-\epsilon$. So, $x_{t,i} \ge 1 - \epsilon$ for $i \in [2]$. Essentially, we have shown that the output profile $\bx_t$ is close to the equilibrium output profile. Lemma~\ref{lm:lipschitz2} below completes the proof by showing that if the output profile is close to the equilibrium profile then the agents cannot benefit much by deviating. 
\begin{lemma}\label{lm:lipschitz2}
If the output profile $\bx = (x_1, x_2)$ satisfies $1 - \epsilon \le x_i \le 1$ for $i \in [2]$, then $\bx$ is an $3\epsilon$-approximate equilibrium.
\end{lemma}
\qed\end{proof}

%% file: 5dynamics.tex
\section{Discounted-Sum Dynamics and its Convergence}\label{sec:dissum}
Let us now look at a dynamic process we call the discounted-sum dynamics and study its convergence properties. We will use these results for our subsequent analysis of the BR dynamics for three or more agents. The discounted-sum dynamics, its convergence, and the potential function used to prove it may be of independent interest.

\begin{definition}[Discounted-Sum Dynamics]\label{def:dissum}
    Let $\bz_0 = (z_{0,i})_{i \in [n]} \in \bR^n$ be the initial state. At time $t \in \bZ_{\ge 0}$, an $i_t \in [n]$ is selected, and the next state is computed as follows:
    \begin{align*}
        z_{t+1, i_t} = - \beta_{t,i}(\bz_t) \sum_{j \neq i_t} z_{t, j}; \qquad z_{t+i, j} = z_{t, j}, \text{ for } j \neq i_t;
    \end{align*}
    where $0 \le \beta_{t,i}(\bz_t) \le B$ for some non-negative constant $B < 1$.
\end{definition}
We assume that $\beta_{t,i}(\bz_t)$ can be chosen arbitrarily (possibly adversarially) but must satisfy $0 \le \beta_{t,i}(\bz_t) \le B$. For convenience, we refer to the $n$ coordinates of the discounted-sum dynamics as agents because these $n$ coordinates will correspond to the $n$ agents when we use this dynamics to analyze the BR dynamics in subsequent sections. We analyze the discounted-sum dynamics for two methods of selecting agent $i_t$---randomized and best-case---similar to the agent selection processes of the BR dynamics as described in Section~\ref{sec:prelim:select}.

We study $\epsilon$-approximate convergence of the discounted-sum dynamics w.r.t. the $\ell_1$-distance, i.e., after sufficiently large $t$, we want $||\bz_t||_1 = \sum_j |z_{t,j}| \le \epsilon$ for any given $\epsilon > 0$. It can be observed that the unique stable point of this dynamics is $\bzero$ (our rate of convergence analysis will also formally imply this claim).
Let $\sigma_t = \sum_j z_{t,j}$ and $\sigma_{t,-i} = \sum_{j \neq i} z_{t,j}$.
\begin{lemma}\label{lm:dissum:random}
    For the randomized agent selection model, the discounted-sum dynamics $\epsilon$-approximately converges to $\bzero$ in $O(\frac{\log(n)}{L^2(1-B)} \log(\frac{f(\bz_0)}{\epsilon \delta}))$ time w.p. $1-\delta$, i.e., $||\bz_t||_1 \le \epsilon$ for all $t \ge T$ and $T = O(\frac{\log(n)}{L^2(1-B)} \log(\frac{f(\bz_0)}{\epsilon \delta}))$ w.p. $1-\delta$.
\end{lemma}
Before we prove Lemma~\ref{lm:dissum:random}, let us briefly discuss the behavior of the discounted-sum dynamics. Say agent $i = i_t$ makes the move at time $t$, then the dynamics updates the $i$-th coordinate as $z_{t+1,i} = - \beta_{t,i}(\bz_t) \sigma_{t,-i}$, where $\beta_{t,i}(\bz_t)$ may be selected adversarially in $[0,B]$. Let's make the following simplifying assumption: say $\beta_{t,i}(\bz_t) = B$ always, so $z_{t+1,i} = - B \sigma_{t,-i} $. With this simplifying assumption, we can construct a potential function $f_{simple}(\bz_t) = \frac{B}{2} \sigma_t^2 + \frac{1-B}{2} || \bz_t ||_2^2$ and prove fast convergence in $O(\frac{1}{L}\log(\frac{n}{\epsilon\delta}))$ time with probability $1-\delta$ as done in \cite{ghosh2023best}. In particular, we can show that $f_{simple}$ is $(1-B)$-strongly convex and $(1+(n-1)B)$-smooth and then use techniques used to analyze coordinate descent~\cite{wright2015coordinate} to get the required convergence rate. 
But for arbitrary $\beta_{t,i}(\bz_t) \in [0,B]$ coefficients, our attempts at constructing smooth and well-behaved potential functions failed (it can be formally shown that no degree-$2$ polynomial can be used as a potential function by constructing suitable examples; this observation likely holds for higher degree polynomials as well). Therefore, we introduce a novel non-smooth \textit{weak} potential function and use it to prove the rate of convergence. The potential either decreases or stays the same (therefore, we call it weak) as the dynamics proceeds.

\input{proofs/dissum-randomA.tex}

Lemma~\ref{lm:dissum:random} proves a high probability bound of $O(\frac{\log(n)}{L^2})$ on the convergence time (w.r.t. $n$ and $L$). We believe that this bound can be improved to $O(\frac{\log(n)}{L})$ using alternate techniques, although unlikely using the potential function $f$ given in equation \eqref{eq:potential} and by measuring progress w.r.t. $f$ only for small time periods, as illustrated in Example~\ref{ex:dynamics:1}. 
\begin{example}\label{ex:dynamics:1}
Let us assume that $\beta_{t,i}(\bz_t) = \frac{1}{2}$ for all $t$, $i$, and $\bz_t$.\footnote{This ensures a $O(\frac{\log(n)}{L})$ rate of convergence using the smooth potential function $f_{simple}$ discussed earlier in the section.} Let $\bz_t = (-1, 1, 0, \ldots)$. Note that $f(\bz_t) = 1$. At time $t$, if we select an agent $i_t \ge 3$, then $i_t$ has $z_{t,i_t} = 0$ by definition and $z_{t+1,i_t} = 0$ because $\sum_{j \neq i_t} z_{t,j} = 1 - 1 = 0$. So, $\bz_{t+1} = \bz_t$, and we make no progress w.r.t. the potential $f$. To make progress, we need to pick $i_t \in \{1,2\}$, which, in the worst case, occurs with a probability of $2L$. So, it easy to check that we will need about $\frac{1}{2L}$ time steps before we pick an agent in $\{1,2\}$. 

Now, after we pick $i_{\tau} \in \{1,2\}$, say $i_{\tau} = 1$ w.l.o.g. at time $\tau \ge t$, we have $\bz_{\tau+1} = (-\frac{1}{2}, 1, 0, \ldots, 0)$. The potential $f(\bz_{\tau+1}) = 1$ still. At time $\tau+1$, if we pick $i_{\tau+1} = 2$, which happens with probability $L$ in the worst case, then $f(\bz_{\tau+2}) = \frac{1}{2}$. But if we pick $i_{\tau+1} \ge 3$, say $i_{\tau+1} = 3$ w.l.o.g., then $\bz_{\tau+2} = (-\frac{1}{2}, 1, -\frac{1}{4}, 0, \ldots, 0)$ and $f(\bz_{\tau+2}) = 1$. And then if we pick $i_{\tau+2} \in \{1,2\}$, then $f(\bz_{\tau+3}) \ge \frac{3}{4}$. Extending this argument, if we pick an agent in $\{1,2\}$ the first time after time $\tau$ at time $\tau+3$, then $f(\bz_{\tau+4}) \ge \frac{7}{8}$. Essentially, the longer the delay in picking an agent in $\{1,2\}$, the less progress we make. By simple calculations, we can check that at time $\tau' > \tau$ when we eventually pick an agent in $\{1,2\}$ the second time, the expected value of the potential is $\Exp[f(\bz_{\tau'+1})] \ge 1 - 2L$. 

To summarize, after about $\frac{2}{L}$ steps, we decreased the potential by a multiplicative factor of about $1 - 2L$. So, if measure only progress w.r.t. $f$ for local time periods, and we do not track any other properties of the state $\bz_t$, then it is unlikely that we can get a convergence time better than $O(\frac{1}{L^2})$.
\end{example}

We next prove almost tight upper and lower bounds for the convergence time for the best-case selection model. By the definition of the best-case model, the lower bound applies to the randomized model as well, which implies that our bound is tight w.r.t. $\epsilon$ and polynomially tight w.r.t. $n$.
\begin{lemma}\label{lm:dissum:bestcase}
    For the best-case selection model, the discounted-sum dynamics $\epsilon$-approximately converges in $O(\frac{n}{1-B} \log(\frac{f(\bz_0)}{\epsilon}))$ time, i.e., $||\bz_t||_1 \le \epsilon$ for all $t \ge T = O(\frac{n}{1-B} \log(\frac{f(\bz_0)}{\epsilon}))$. Further, for the best-case selection model, the dynamics may take $\Omega(n + \frac{1}{1-B}\log(\frac{f(\bz_0)}{\epsilon}))$ time to $\epsilon$-approximately converge, i.e., $||\bz_t||_1 > \epsilon$ for all $t < T = \Omega(n + \frac{1}{1-B}\log(\frac{f(\bz_0)}{\epsilon}))$ for some starting points $\bz_0$.
\end{lemma}

%% file: proofs/dissum-randomA.tex
\begin{proof}[Lemma~\ref{lm:dissum:random}, partial]
Given an $n$-dimensional vector $\bz \in \bR^n$, we define the following potential function
\begin{equation}\label{eq:potential}
    f(\bz) = \max\left( \sum_{j \in [n]} \bone(z_{j} > 0) z_{j}, - \sum_{j \in [n]} \bone(z_{j} < 0) z_{j} \right),
\end{equation}
where $\bone$ is the indicator function. The potential function $f$ separates the positive and negative parts of $\bz$ and takes the max of the sums of their absolute values. Notice that $f(\bz) \ge 0$ for all $\bz \in \bR^n$ and $f(\bz) = 0 \Longleftrightarrow \bz = \bzero$.

Let $v_{t,j} = \bone(z_{t,j} > 0) z_{t,j}$ and $w_{t,j} = - \bone(z_{t,j} \le 0) z_{t,j}$. Let $V_t = \sum_j v_{t,j}$ and $W_t = \sum_j w_{t,j}$. Let also $\cV_t = \{ j \in [n] \mid z_{t,j} > 0 \}$ and $\cW_t = \{ j \in [n] \mid z_{t,j} \le 0 \}$. Notice that $\cV_t \cap \cW_t = \emptyset$ and $\cV_t \cap \cW_t = [n]$. Also notice that $f(\bz_t) = \max(V_t, W_t)$.

First, let us prove that the potential never increases. Let us assume $V_t \ge W_t$; this is w.l.o.g. because the dynamics and the potential function are symmetric w.r.t. the positive and the negative parts of $\bz_t$.\footnote{An agent $j$ with $z_{t,j} = 0$ can be put on either the positive or the negative side; as a convention, we associate them to the negative side, which causes slight asymmetry. It can be checked that our proof works irrespective of the side we assign to these agents with $z_{t,j} = 0$.} Say agent $i$ is picked at time $t$, we have the following cases:
\begin{itemize}
    \item $i \in \cV_t$. Depending upon the value of $v_{t,i}$, we have the following two sub-cases:
    \begin{itemize}
        \item $W_t \le V_t - v_{t,i}$. Then agent $i$ will change sign after this move: $v_{t+1,i} = 0$ and $w_{t+1,i} = \beta_{t,i}(\bz_t) (V_t - v_{t,i} - W_t) \le B (V_t - v_{t,i} - W_t)$. So, the updated value of the potential is $f(\bz_{t+1}) = \max(V_t - v_{t,i}, W_t + w_{t+1,i}) \le \max(V_t - v_{t,i}, (1-B)W_t + B (V_t - v_{t,i})) = V_t - v_{t,i} \le V_t = \max(V_t, W_t) = f(\bz_t)$.
        
        \item $W_t > V_t - v_{t,i}$. Then agent $i$ will keep the same sign after the move: $v_{t+1,i} \le B (W_t - (V_t - v_{t,i}))$. So, the updated value of the potential is $f(\bz_{t+1}) = \max(V_t - v_{t,i} + v_{t+1,i}, W_t) \le \max( (1 - B) (V_t - v_{t,i}) + B W_t, W_t) = W_t \le \max(V_t, W_t) = f(\bz_t)$.
    \end{itemize}

    \item $i \in \cW_t$. As $W_t \le V_t$, the agent $i$ will keep the same sign after the move: $w_{t+1,i} \le B (V_t - (W_t - w_{t,i}))$. So, the updated value of the potential is $f(\bz_{t+1}) = \max(V_t, W_t - w_{t,i} + w_{t+1,i}) \le \max(V_t, B V_t + (1-B) (W_t - w_{t,i})) = V_t = f(\bz_t)$.
\end{itemize}
We continue with the rate of convergence analysis in the appendix.
\qed\end{proof}

%% file: 6hom-n.tex
\section{Three or More Homogeneous Agents}\label{sec:hom_n_ub}
We now study BR dynamics for $n \ge 3$ agents. As introduced in Section~\ref{sec:prelim:select}, we consider a randomized model and a best-case deterministic model for selecting the agent who makes the transition at any given time point $t$. 

For the randomized model, agent $i$ makes the transition at time $t$ w.p. $p_{t,i}(\bx_t)$ given action profile at $t$ is $\bx_t$. We assume $p_{t,i}(\bx_t) \ge L > 0$ for all agents $i$ except the one who played at time $t-1$. An important special case of our model is to assume that $p_{t,i}(\bx_t) = 1/n$ for all $t$, $i$, and $\bx_t$, i.e., the agent making the move is selected uniformly at random. We do a worst-case analysis over all $p_{t,i}$ given the parameter $L > 0$. 

For the best-case deterministic model, we assume that the agents are selected to ensure the fastest possible convergence. We provide a lower bound on convergence time for this model, which automatically carries over to the randomized model.

\begin{theorem}\label{thm:hom}
BR dynamics in lottery contests with $n \ge 3$ homogeneous agents and randomized selection of agents reaches an $\epsilon$-approximate equilibrium w.p. $1-\delta$ in 
$O ( \frac{1}{nL} \log\log(\frac{1}{\gamma}) + \frac{\log(n)}{L^2} \log(\frac{nK}{\epsilon \delta}) )$
steps for every $\epsilon, \delta \in (0,1)$, where $\gamma$ is a function of the initial output profile as given in Lemma~\ref{lm:outputlb}, and assuming that the cost function $c$ is $K$-Lipschitz continuous in the interval $[1-\epsilon,1+\epsilon]$. 
\end{theorem}
Notice that the above theorem implies that when agents are selected uniformly at random at each time step (i.e., $L = 1/n$), the time till convergence is bounded above by 
$O ( \log\log(\frac{1}{\gamma}) + n^2\log(n) \log(\frac{nK}{\epsilon \delta}) )$ 
w.p. $1-\delta$.

\begin{proof}[Theorem~\ref{thm:hom}]
We shall use an extension of the well-known coupon collector problem in the first part of our analysis. This portion of the analysis is an extension of a similar analysis done in \cite{ghosh2023best}. In the coupon collector problem, we select a coupon out of $n$ coupons uniformly at random at each time step, and we want to bound the time it takes to collect all coupons. A tight $\Theta(n \log n)$ high probability bound is known for this problem (see, e.g., \cite{mitzenmacher2017probability}), which can be used to derive the following bounds on the time it takes for each agent to play at least once in the BR dynamics.

\begin{lemma}\label{lm:coupon}
    Let $T$ be the time it takes for all agents to play at least once in the BR dynamics, we have the following high probability bounds: (i) upper bound, $\Prob[T \le \frac{1}{L}\ln (\frac{n}{\delta})] \ge 1 - \delta$; (ii) lower bound, $\Prob[T \ge \frac{1}{L}\log(n \delta)] \ge 1 - \delta$.
\end{lemma}

Keeping the same notation as our previous discussions, let $x_{t,i}$ be the output of agent $i \in [n]$ at time $t \ge 0$. Note that $x_{t,i}$ is a random variable and $(x_t)_{t \ge 0}$ is a stochastic process because the agent $i_t$ making the move at time $t$ is selected randomly. As before, let $s_t = \sum_{j} x_{t,j}$ and $s_{t,-i} = \sum_{j \neq i} x_{t,j}$; $s_t$ and $s_{t,-i}$ are also stochastic processes.

We call an initial time period of the BR dynamics the \textit{warm-up} phase (Definition~\ref{def:warmup}).
When the warm-up phase ends, we ensure that the output profile is, and stays thereafter, in a certain well-behaved region where we can avoid the corner case of BR to $0$ output.
\begin{definition}[Warm-Up Phase]\label{def:warmup}
    The time period $\{0, 1, \ldots, T_{warm}-1\}$ denotes the warm-up phase, where $T_{warm}$ is the smallest time such that for every $t \ge T_{warm}$:
    \begin{enumerate}
        \item \label{def:warmup:3} the total output is strictly less than $\frac{n^2}{(n-1) c'(0)}$: $s_t < \frac{n^2}{(n-1) c'(0)}$;
        \item \label{def:warmup:1} all agents produce output of at most $\frac{n^2}{4(n-1)}$: $0 \le x_{t,i} \le \frac{n^2}{4(n-1)}$ for every $i$;
        \item \label{def:warmup:2} there are at least two agents $i$ and $j \neq i$ with positive output: $x_{t,i} > 0$ and $x_{t,j} > 0$.
    \end{enumerate}
\end{definition}
In the first part of the analysis, we bound the time it takes for the warm-up phase to finish with high probability.


\paragraph{\textbf{Analysis Part 1 (Warm-Up Phase).}}
We shall use the following properties (Lemma~\ref{lm:prop1}) of BR dynamics in our subsequent analysis.
\begin{lemma}\label{lm:prop1}
    The BR dynamics satisfies the following properties:
    \begin{enumerate}
        \item \label{lm:prop1:1} $s_t > 0$ for every $t \ge 1$.
        \item \label{lm:prop1:2} $s_t < \frac{n^2}{(n-1) c'(0)} \implies s_{t+1} < \frac{n^2}{(n-1) c'(0)}$ for every $t \ge 0$.
        \item \label{lm:prop1:3} For $i,j \in [n]$ such that $i \neq j$, if $x_{t,i} > 0$, $x_{t,j} > 0$, and $s_t < \frac{n^2}{(n-1) c'(0)}$, then $x_{t+1,i_t} > 0$ for every $t \ge 0$.
    \end{enumerate}
\end{lemma}

The next lemma gives a high probability bound on the time taken by the warm-up phase.
\begin{lemma}\label{lm:part1}
    Time till completion of the warm-up phase $T_{warm} = O(\frac{1}{L}\log(\frac{n}{\delta}))$ w.p. $1 - \delta$.
\end{lemma}

In Lemma~\ref{lm:part1}, we proved that the warm-up phase finishes with high probability in about $O(\frac{1}{L}\log n)$ steps. This completes the first part of our analysis.

\paragraph{\textbf{Analysis Part 2 (Reducing to Discounted-Sum Dynamics).}}
We will now reduce the BR dynamics to the discounted-sum dynamics discussed in Section~\ref{sec:dissum}.
We assume that all conditions for the completion of the warm-up phase (Definition~\ref{def:warmup}) are satisfied.

Let $z_{t,i} = x_{t,i} - 1$. We will show that under certain conditions, $z_{t,i}$ follows the discounted-sum dynamics with suitably selected parameter $B \in [0,1)$ (Definition~\ref{def:dissum}). To do that, we need to show that $z_{t+1,i} = - \beta_{t,i}(\bz_t) \sum_{j \neq i} z_{t,j}$ for some suitable $\beta_{t,i}(\bz_t) \in [0,B]$ for all $t$, $i$, $\bz_t$. Let $\sigma_t = \sum_j z_{t,j} = s_t - n$ and $\sigma_{t, -i} = \sum_{j \neq i} z_{t,j} = s_{t,-i} - (n-1)$.
\begin{lemma}\label{lm:good_domain1}
    If $s_{t, -i_t} \ge \frac{1}{n-1}$, then $z_{t+1, i_t} = x_{t+1, i_t} - 1 = - \beta_{t,i}(\bx_t) (s_{t, -i_t} - (n-1)) = - \beta_{t,i}(\bx_t) \sigma_{t, -i_t}$, where $\beta_{t,i}(\bx_t) \in [0, B]$ and $B \le \frac{1}{2}$.
\end{lemma}
Lemma~\ref{lm:good_domain1} proves that the BR dynamics behaves like the discounted-sum dynamics if $s_{t, -i_t} \ge \frac{1}{n-1}$. In the next few lemmas, we try to bound the time it takes to satisfy this condition with high probability. 
\begin{lemma}\label{lm:good_domain2}
    (Assuming the completion of the warm-up phase.) If at time $t$, the output profile $\bx_t$ satisfies the following condition, then $\bx_{\tau}$ also satisfies it for all $\tau \ge t$. Condition: There are at least two agents $i$ and $j \neq i$ such that $x_{t,i} \ge \frac{1}{n-1}$ and $x_{t,j} \ge \frac{1}{n-1}$.
\end{lemma}
Notice that, if the condition in Lemma~\ref{lm:good_domain2} is satisfies by $\bx_t$, then there are two agents $i$ and $j$ with $x_{t,i} \ge \frac{1}{n-1}$ and $x_{t,j} \ge \frac{1}{n-1}$, which implies that for every $k \in [n]$, $s_{t, -k} \ge \min(x_{t,i}, x_{t,j}) \ge \frac{1}{n-1}$ and Lemma~\ref{lm:good_domain1} ensures that the BR dynamics resembles a discounted-sum dynamics. We next bound the time it takes to satisfy the condition in Lemma~\ref{lm:good_domain2}.

\begin{lemma}\label{lm:good_domain3}
    For every $t \ge T_{warm} + T_{sum}$, there are at least two agents $i$ and $j \neq i$ such that $x_{t,i} \ge \frac{1}{n-1}$ and $x_{t,j} \ge \frac{1}{n-1}$, where $T_{sum}$ is defined as: 
    \begin{itemize}
        \item $T_{sum} = O(\frac{1}{n L} \ln(\frac{1}{\delta}))$ w.p. $1-\delta$ if $\gamma = \min_{j \in [n]} s_{T_{warm}, -j} \ge \frac{1}{n-1}$;
        \item $T_{sum} = O(\frac{1}{n^2 L^2} \ln(\frac{1}{\delta}) + \frac{1}{nL} \log\log(\frac{1}{\gamma}) )$ w.p. $1-\delta$ if $\gamma = \min_{j \in [n]} s_{T_{warm}, -j} < \frac{1}{n-1}$.
    \end{itemize}     
\end{lemma}
Notice that the duration $T_{sum}$ given in Lemma~\ref{lm:good_domain3} depends upon $\gamma = \min_{j \in [n]} s_{T_{warm}, -j}$ when $\gamma < \frac{1}{n-1}$. $\gamma$ is defined as a function of the output profile $\bx_{T_{warm}}$ at time $T_{warm}$, which may be considered unsatisfactory, so we next try to lower bound $\gamma$ as a function of the initial output profile $\bx_0$.

\begin{lemma}\label{lm:outputlb}
Let $\gamma = \min_{j \in [n]} s_{T_{warm}, -j}$. If $\gamma < \frac{1}{n-1}$, then it is lower bounded as $\gamma \ge \gamma_{lb} = \min(\{a\} \cup \cA \cup \cB )$, where $\cA = \{ x_{0,j} \mid x_{0,j} > 0, j \in [n] \}$ and $\cB$ is either (1) $\cB = \{ BR(x_{0,j} + 1) \mid j \in [n] \}$ if $c'(0) = 0$ or (2) $\cB = \{ \min(\frac{\kappa - x_{0,j}}{4}, BR( \frac{\kappa + x_{0,j}}{4} ) ) \mid x_{0,j} < \kappa, j \in [n] \}$ if $c'(0) > 0$, where $\kappa = \frac{n^2}{c'(0) (n-1)}$.
\end{lemma}
Lemma~\ref{lm:outputlb} intuitively says that the value of $\gamma$ at the end of the warm-up phase is lower bounded by either $a$, or the smallest positive output by a player at time $t=0$, or is close to the smallest positive BR to the largest output (that has a positive BR) by any player at time $t=0$.

Let us summarize our results till now. In Lemma~\ref{lm:part1}, we showed that the time taken by the warm-up phase is $O(\frac{1}{L}\log(\frac{n}{\delta'}))$ w.p. $1-\delta'$. Then, assuming the completion of the warm-up phase, we have a $1-\delta'$ high-probability bound of $O(\frac{1}{n^2 L^2} \ln(\frac{1}{\delta'}) + \frac{1}{nL} \log\log(\frac{1}{\gamma}) )$ on the time taken to reach a phase of the BR dynamics that corresponds to the discounted-sum dynamics (Lemma~\ref{lm:good_domain2}). 

Next, we use our results in Section~\ref{sec:dissum} to get a high-probability bound on the time it takes for the output profile to reach close to the equilibrium profile. In particular, as $z_{t,i} = x_{t,i} - 1 \ge -1$ and $z_{t,i} = x_{t,i} - 1 \le \frac{n^2}{4(n-1)}$ for all $t \ge T_{warm}$, so $\sum_{i} |z_{t,i}| \le \frac{n^3}{4(n-1)} \le n^2$. 
Using Lemma~\ref{lm:dissum:random}, we get $|| \bz_{\tau + T_{warm} + T_{sum}} ||_1 \le \epsilon$ for all $\tau \ge T = O(\frac{\log(n)}{L^2} \log(\frac{n}{\epsilon \delta'}))$ w.p. $1-\delta'$.
Finally, using Lemma~\ref{lm:lipschitz} below, we show that small $\ell_1$-distance implies approximate equilibria. Setting $\delta' = \delta/3$ and using union bound on the total probability of failure  completes our analysis.
\begin{lemma}\label{lm:lipschitz}
Given an output profile $\bx = (x_i)_{i \in [n]}$ and the equilibrium profile $\bx^* = (1, \ldots, 1)$, if $|| \bx - \bx^* ||_1 \le \epsilon$, $| BR(s_{-i}) - 1 | \le \epsilon$ for all $i \in [n]$, and the (homogeneous) cost function $c$ is $K$-Lipschitz continuous in the interval $[1-\epsilon,1+\epsilon]$, then $\bx$ is an $4n\epsilon$-approximate equilibrium for $\epsilon \le \frac{1}{4n}$.
\end{lemma}
\qed\end{proof}

We next prove almost tight upper and lower bounds for the time required for the convergence of the best-case selection model.
\begin{theorem}\label{thm:hom:bestcase}
BR dynamics in lottery contests with $n \ge 3$ homogeneous agents and best-case selection of agents reaches an $\epsilon$-approximate equilibrium in 
$O( \log\log(\frac{1}{\gamma}) + n \log(\frac{nK}{\epsilon}) )$
steps for every $\epsilon \in (0,1)$, where $\gamma$ is a function of the initial output profile as given in Lemma~\ref{lm:outputlb}, and assuming that the cost function $c$ is $K$-Lipschitz continuous in the interval $[1-\epsilon,1+\epsilon]$. 
As a lower bound, we show that the convergence takes at least 
$\Omega(n + \log\log(\frac{1}{\gamma}) + \frac{1}{\log(n)}\log(\frac{1}{\epsilon}))$ steps.
\end{theorem}
The lower bound for the best-case selection model automatically applies to the randomized selection model; the dependency of $\Omega(n)$ given in Theorem~\ref{thm:hom:bestcase} can be slightly improved to $\Omega(\frac{1}{L} \log(n \delta) )$ w.p. $1-\delta$ using Lemma~\ref{lm:coupon}. Note that these bounds are in the worst case over all convex cost functions. In particular, except for the lower bound of $\Omega(n)$ that holds for all cost functions, for every $\epsilon$ and $\gamma$, there exists a convex cost function that reaches an approximate equilibrium after just one BR transition by each agent. E.g., if $c'(z) = z^r$ for all $z \ge 0$ and $r \rightarrow \infty$, then $BR(s_{-i}) \rightarrow 1$ for any $i$ and $s_{-i} < \infty$.
The lower bounds w.r.t. $\gamma$ and $\epsilon$ have been proven using the linear cost function: $c'(z) = 1$ for all $z \ge 0$. 
We defer tighter analyses of lower and upper bounds for specific classes of convex cost functions, as done for the two-agent case in Theorem~\ref{thm:hom2:curved}, for future work.

%% file: 9conclusion.tex
\section{Conclusion and Open Problems}
We showed fast convergence of BR dynamics in Tullock contests with homogeneous agents and convex cost functions. In particular, we introduced and analyzed the discounted-sum dynamics and used it in the analysis of the BR dynamics.
The following are a few open problems: (1) Our upper bound for the discounted-sum dynamics has a $\Tilde{O}(n^2)$ dependency on $n$, but our simulations indicate a $\Tilde{O}(n)$ dependency. (2) The BR dynamics does not converge for non-homogeneous agents, but slightly more sophisticated learning dynamics may converge. (3) Simulations show convergence of BR dynamics for almost homogeneous agents. (4) Prove that strong convexity of the cost function helps improve the convergence rate of the BR dynamics and makes it more robust to small non-homogeneity of agents. 

%% file: 10appendix.tex
\input{7non-hom.tex}

\section{Proofs from Section~\ref{sec:hom_2_agents}} 

\input{proofs/2lm-le1.tex}
\input{proofs/2hom-1.tex}
\input{proofs/2hom-2.tex}
\input{proofs/2lipschitz.tex}
\input{proofs/2hom-curved.tex}

\section{Proofs from Section~\ref{sec:dissum}}
\input{proofs/dissum-randomB.tex}
\input{proofs/partition.tex}
\input{proofs/dissum-bestcase.tex}

\section{Proofs from Section~\ref{sec:hom_n_ub}}
\input{proofs/coupon.tex}
\input{proofs/prop1.tex}
\input{proofs/part1.tex}

\input{proofs/good_domain1.tex}
\input{proofs/good_domain2.tex}
\input{proofs/good_domain3.tex}
\input{proofs/outputlb.tex}
\input{proofs/lipschitz.tex}

\input{proofs/hom-bestcase.tex}

%% file: 7non-hom.tex
\section{Non-Homogeneous Agents}\label{sec:non-hom}
Ghosh and Goldberg~\cite{ghosh2023best} provided examples with non-homogeneous agents and linear cost functions where the BR dynamics cycles and does not converge. They also show that the non-convergence result is generic, i.e., there are examples where the set of starting points that leads to non-convergence is of positive measure. But, all their examples have cycles that pass through a transition where an agent has to best respond to $0$ total output by other agents. As the BR to $0$ is not well defined, these examples are somewhat unsatisfactory. We complement these results by providing an example with strictly convex cost functions where the cycle does not pass through a BR to $0$.
\begin{example}\label{ex:nonhom:1}
Let $c_1(z) = z^{1.2}$ and $c_2(z) = z^{1.2}/20$ for $z \ge 0$. BR dynamics leads to the cycle given in Table~\ref{tab:nonhom}.
\begin{table}
    \centering
    \begin{tabular}{| c | c | c | c | c | c | c | c }
        \hline
        $x_{t,1}$ & $0.1058$ & $0.1131$ & $0.1131$ & $0.1058$ & $0.1058$ & $0.1131$ & $\ldots$ \\
        \hline
        $x_{t,2}$ & $1.3102$ & $1.3102$ & $1.3468$ & $1.3468$ & $1.3102$ & $1.3102$ & $\ldots$ \\
        \hline
    \end{tabular}
    \caption{Non-convergence of BR dynamics for non-homogeneous agents.}
    \label{tab:nonhom}
\end{table}
\end{example}
In the cycle presented in Example~\ref{ex:nonhom:1}, when the weaker agent (high cost $c_1(z) = z^{1.2}$) produces low output, then the stronger agent also produces relatively low output. But then the weaker agent tries to compete by producing higher output, then the stronger agent responds with a relatively higher output, and then the weaker agent backs off and produces low output. And this cycle continues.

The initial action profiles that lead to the cycle in Example~\ref{ex:nonhom:1} are generic. The initial profile can also be close to $0$, e.g., starting from $(*,10^{-5})$ and agent $1$ making the first move leads to the cycle presented above. Like \cite{ghosh2023best}, we also observe that BR dynamics for \textit{almost} homogeneous agents always converges in our simulations. We also observe that \textit{stronger} convexity leads to faster and more robust convergence results, i.e., when the cost functions are \textit{more} convex, non-convergence becomes unlikely and the rate of convergence improves. Our intuition for these observations is the fact that convexity of the cost function causes an agent to play closer to the equilibrium output, i.e., given the same total output by other agents, $s_{-i}$, the BR of agent $i$, $BR(s_{-i})$, is closer to the equilibrium output if the cost function is more convex. (A version of this behavior can be derived using the first-order condition, equation~\eqref{eq:br}, and has been implicitly used frequently in our convergence analysis.) We defer the formal analysis of these aspects of the BR dynamics to future work.

%% file: proofs/2lm-le1.tex
\begin{proof}[Lemma~\ref{lm:<=1}]
    For cleaner exposition, let $v = x_{-i} = s_{-i}$ and $y = BR(x_{-i}) = BR(v)$. If $v = 0$, then $y = BR(v) = a < 1$ by definition. Let us assume $v > 0$. From the first order condition on the utility function, equation \eqref{eq:br}, we have
    \begin{equation}\label{eq:lm:<=1:1}
        \left. \frac{\partial u_i(z, v)}{\partial z} \right|_{z = y = BR(v)} = 0 \implies \frac{v}{(y + v)^2} - \frac{1}{4} c'(y) = 0.
    \end{equation}
    We know that $u_i(z, v)$ is a strictly concave function of $z$, or, in other words, $\frac{\partial u_i(z, v)}{\partial z}$ is a strictly decreasing function of $z$. Substituting $z$ by $v$ in $\frac{\partial u_i(z, v)}{\partial z}$, we get
    \[
        \left. \frac{\partial u_i(z, v)}{\partial z} \right|_{z = v} = \frac{v }{4 v^2} - \frac{c'(v)}{4} = \frac{1}{4v} \left( 1 - v c'(v) \right).
    \]
    If $v < 1$, then $\left. \frac{\partial u_i(z, v)}{\partial z} \right|_{z = v} = 1 - v c'(v) > 0$ because $c'(v) < c'(1) = 1$. And, as $\frac{\partial u_i(z, v)}{\partial z}$ is strictly decreasing in $z$, 
    we have $y > v$. With a similar argument, if $v = 1$, then $y = 1$, and if $v > 1$, then $y < v$.

    Notice that equation~\eqref{eq:lm:<=1:1} implicitly defines $y$ as a function of $v$. Let us now compute the derivative of $y$ w.r.t. $v$ by differentiating equation~\eqref{eq:lm:<=1:1} w.r.t. $v$,
    \begin{align*}
        &\frac{d}{d v} \left( \frac{v}{(y + v)^2} - \frac{1}{4} c'(y) \right) = 0 \\
        & \qquad \implies \frac{1}{(y + v)^2} - \frac{2v}{(y + v)^3}  + \frac{dy}{dv} \left( \frac{-2 v}{(y + v)^3} - \frac{1}{4} c''(y) \right) = 0 \\
        & \qquad \implies \frac{(y+v) - 2v}{(y + v)^3} = \frac{dy}{dv} \frac{1}{(y+v)^3} \left( 2v + \frac{(y+v)^3}{4} c''(y) \right) \\
        & \qquad \implies \frac{dy}{dv} = \frac{y - v}{2v + \frac{(y+v)^3}{4} c''(y)}.
    \end{align*}
    As $c''(y) \ge 0$, if $y \ge v$ then $\frac{d y}{d v} \ge 0$. We showed earlier that $v \le 1 \Longleftrightarrow y \ge v$, so $v \le 1 \implies \frac{d y}{d v} \ge 0$. With a similar argument, $v \ge 1 \Longleftrightarrow y \le v \implies \frac{d y}{d v} \le 0$. Therefore, the maximum value of $y$ occurs at $v = 1$, where $y$ is also $1$. Same steps also show that, if $v < 1$, then $y > v$, then $\frac{d y}{d v} > 0$, which implies $y < 1$.
\qed\end{proof}

%% file: proofs/2hom-1.tex
\begin{proof}[Lemma~\ref{lm:2hom:1}]
    We know from the first-order condition that $\frac{z_t}{(z_t + z_{t+1})^2} = \frac{1}{4} c'(z_{t+1})$. As $z_{t+1} < 1 \implies c'(z_{t+1}) < 1$ and $z_{t+1} > z_t$ for every $t$, we have
    \begin{align*}
        &\frac{z_t}{(z_t + z_{t+1})^2} = \frac{c'(z_{t+1})}{4} \implies \frac{z_t}{(2 z_{t+1})^2} \le \frac{1}{4} \implies z_{t+1} \ge \sqrt{z_t} \implies z_{t} \ge z_0^{\frac{1}{2^t}}.
    \end{align*}
    Now, we want $z_t \ge \frac{1}{2}$. So, if we ensure that $z_0^{\frac{1}{2^t}} = \gamma^{\frac{1}{2^t}} \ge \frac{1}{2}$, then we are done.
    \begin{equation*}
        \gamma^{\frac{1}{2^t}} \ge \frac{1}{2} \Longleftrightarrow \frac{1}{2^t} \lg\left( \gamma \right) \ge -1 \Longleftrightarrow \lg\left( \frac{1}{\gamma} \right) \le 2^t \Longleftrightarrow t \ge \lg\lg\left(\frac{1}{\gamma} \right).
    \end{equation*}
\qed\end{proof}

%% file: proofs/2hom-2.tex
\begin{proof}[Lemma~\ref{lm:2hom:2}]
    As $z_{t+1} < 1$ for every $t$, which implies $c'(z_{t+1}) < 1$, we have
    \begin{align*}
        &\frac{z_t}{(z_t + z_{t+1})^2} = \frac{c'(z_{t+1})}{4} \implies 4 z_t \le (z_t + z_{t+1})^2 \implies z_{t+1} \ge \sqrt{z_t}(2 - \sqrt{z_t}).
    \end{align*}
    Let $\zeta_t = 1 - \sqrt{z_t} \Longleftrightarrow z_t = (1 - \zeta_t)^2$. From the above inequality, we have
    \begin{align*}
        z_{t+1} \ge \sqrt{z_t}(2 - \sqrt{z_t}) &\Longleftrightarrow (1 - \zeta_{t+1})^2 \ge (1 - \zeta_t) (1 + \zeta_t) = 1 - \zeta_t^2 \\
        &\Longleftrightarrow 1 + \zeta_{t+1}^2 - 2 \zeta_{t+1} \ge 1 - \zeta_t^2 \Longleftrightarrow \zeta_{t+1} \le ( \zeta_{t+1}^2 + \zeta_{t}^2)/2.
    \end{align*}
    As $z_{t} < z_{t+1}$, so $\zeta_t > \zeta_{t+1}$, which implies
    \begin{align*}
        \zeta_{t+1} \le ( \zeta_{t+1}^2 + \zeta_{t}^2)/2 \implies \zeta_{t+1} \le \zeta_{t}^2 \implies \zeta_{t} \le \zeta_{0}^{2^t}.
    \end{align*}
    From the initial condition, we know that $z_0 \ge \frac{1}{2} \implies \zeta_0 \le 1 - \frac{1}{\sqrt{2}} \le \frac{1}{2}$. We want $z_t \ge 1 - \epsilon \impliedby \sqrt{z_t} \ge \sqrt{1 - \epsilon}$. Using standard inequalities, $\sqrt{1 - \epsilon} \le \sqrt{e^{-\epsilon}} = e^{-\epsilon/2} \le 1 - \epsilon/4$ for every $\epsilon \le 1$. So, $\sqrt{z_t} \ge \sqrt{1 - \epsilon}$ is implied by $\sqrt{z_t} \ge 1 - \epsilon/4$. Putting everything together,
    \begin{multline*}
        z_t \ge 1 - \epsilon \impliedby \sqrt{z_t} \ge 1 - \frac{\epsilon}{4} \Longleftrightarrow \zeta_t \le \frac{\epsilon}{4} \impliedby \zeta_{0}^{2^t} \le \frac{\epsilon}{4} \Longleftrightarrow 2^t \lg(\zeta_0) \le \lg(\epsilon) - 2 \\
        \Longleftrightarrow 2^t \lg\left(\frac{1}{\zeta_0}\right) \ge 2 + \lg\left(\frac{1}{\epsilon}\right) 
        \impliedby 2^t \ge 3\lg\left(\frac{1}{\epsilon}\right) \impliedby t \ge 2 + \lg\lg\left(\frac{1}{\epsilon}\right).
    \end{multline*}
\qed\end{proof}

%% file: proofs/2lipschitz.tex
\begin{proof}[Lemma~\ref{lm:lipschitz2}] 
Fix an arbitrary agent $i$. 
Let $u_- = u_i(x_i, x_{-i})$ and $u_+ = u_i(BR(x_{-i}), x_{-i})$.
We want to prove that
\begin{align*}
    u_- \ge (1 - 3\epsilon) u_+ \Longleftrightarrow \frac{u_-}{u_+} \ge 1 - 3\epsilon.
\end{align*}
We know that $1-\epsilon \le x_i \le 1$ and $1-\epsilon \le x_{-i} \le 1$, and from Lemma~\ref{lm:<=1}, we also know that $1-\epsilon \le BR(x_{-i}) \le 1$. Further, as $c$ is continuous and $c'(z) \le 1$ for $z \le 1$, we have $c(z) \ge c(1) - \epsilon$ for $z \ge 1-\epsilon$.
So, we can lower bound $u_- $ as
\begin{align*}
    u_- &= u_i(x_i, x_{-i}) = \frac{x_i}{x_i + x_{-i}} - \frac{1}{4} c(x_i) \ge \min_{ y, z \in [1-\epsilon, 1] } \left( \frac{y}{y + z} - \frac{1}{4} c(y) \right) \\
    \ge& \min_{ y, z \in [1-\epsilon, 1] } \frac{y}{y + z} - \max_{ y \in [1-\epsilon, 1] } \frac{1}{4} c(y) \ge \min_{ y \in [1-\epsilon, 1] } \frac{y}{y + 1} - \frac{1}{4} c(1) \ge \frac{1 - \epsilon}{2-\epsilon} - \frac{1}{4} c(1).
\end{align*}
Similarly, we can upper bound $u_+ $ as
\begin{align*}
    u_+ &= u_i(BR(x_{-i}), x_{-i}) = \frac{BR(x_{-i})}{BR(x_{-i}) + x_{-i}} - \frac{1}{4} c(BR(x_{-i})) \\
    &\le \max_{ y, z \in [1-\epsilon, 1] } \left( \frac{y}{y + z} - \frac{1}{4} c(y) \right) \le \max_{ y, z \in [1-\epsilon, 1] } \frac{y}{y + z} - \min_{ y \in [1-\epsilon, 1] } \frac{1}{4} c(y)  \\
    &\le \max_{ y \in [1-\epsilon, 1] } \frac{y}{y + 1 - \epsilon} - \frac{1}{4} c(1-\epsilon) \le \frac{1}{2-\epsilon} - \frac{1}{4} (c(1) - \epsilon)
\end{align*}
Putting these two bounds together, we get
\begin{align*}
    \frac{u_-}{u_+} &\ge \frac{ \frac{1 - \epsilon}{2-\epsilon} - \frac{1}{4} c(1) }{ \frac{1}{2-\epsilon} - \frac{1}{4} (c(1) - \epsilon) } = \frac{ 1 - \frac{2-\epsilon}{4} c(1) - \epsilon }{ 1 - \frac{2-\epsilon}{4} c(1) + \frac{2-\epsilon}{4} \epsilon } \ge \frac{ 1 - \frac{1}{2} c(1) - \epsilon }{ 1 - \frac{1}{2} c(1) + \frac{1}{2} \epsilon } \\
    &\ge \frac{1 - 2\epsilon}{1 + \epsilon}, \text{ because $c(1) \le 1$ as $c(0) = 0$ and $c'(z) \le 1$ for $z \le 1$,}  \\
    &\ge (1 - 2\epsilon) (1 - \epsilon) \ge 1 - 3\epsilon, \text{ because $\frac{1}{1+z} \ge 1-z$ for $0 \le z \le 1$.}
\end{align*}
\qed\end{proof}

%% file: proofs/2hom-curved.tex
\begin{proof}[Theorem~\ref{thm:hom2:curved}]
The proof is similar to the proof of Theorem~\ref{thm:hom2}.
We define the sequence $(z_t)_{t \ge 0}$ the same way as Theorem~\ref{thm:hom2}.
We can use the analysis for Theorem~\ref{thm:hom2} until Lemma~\ref{lm:2hom:1}; then we prove stronger versions of Lemmas~\ref{lm:2hom:1}~and~\ref{lm:2hom:2} assuming $c'(x_i) \in [x_i^p, x_i^q]$, given in Lemmas~\ref{lm:2hom:curved:1}~and~\ref{lm:2hom:curved:2}, respectively. We can also simplify the value of $\gamma$: if $z_0 = x_{0,1} > 1$ and $c'(0) < \frac{4}{z_0}$, then $\gamma = (c')^{-1} \left( \frac{1}{z_{0}} \right) = \left( \frac{1}{z_0} \right)^{1/r}$ for some $r \in [q,p]$. 
\begin{lemma}\label{lm:2hom:curved:1}
    Given $z_0 = \gamma \in (0,\frac{1}{2}]$, $z_t \ge \frac{1}{2}$ for all $t \ge \frac{1}{\lg(2+q)}\lg\lg\left(\frac{1}{\gamma} \right)$ and $z_t < \frac{1}{2}$ for all $t < \frac{1}{\lg(2+p)}\lg\lg\left(\frac{1}{\gamma} \right) - O(1)$.
\end{lemma}
\begin{proof}
    We know from the first-order condition that $\frac{z_t}{(z_t + z_{t+1})^2} = \frac{1}{4} c'(z_{t+1})$. As $c'(z_{t+1}) \le z_{t+1}^q$ and $z_{t+1} > z_t$ for every $t$,
    \begin{multline*}
        \frac{z_t}{(z_t + z_{t+1})^2} = \frac{c'(z_{t+1})}{4} \implies \frac{z_t}{(2 z_{t+1})^2} \le \frac{z_{t+1}^q}{4} \implies z_{t+1} \ge z_t^{\frac{1}{2+q}}  \\
        \implies z_{t} \ge z_0^{\frac{1}{(2+q)^t}} = \gamma^{\frac{1}{(2+q)^t}}.
    \end{multline*}
    Now, we want $z_t \ge \frac{1}{2}$. So, if we ensure that $\gamma^{\frac{1}{(2+q)^t}} \ge \frac{1}{2}$, then we are done.
    \begin{multline*}
        \gamma^{\frac{1}{(2+q)^t}} \ge \frac{1}{2} \Longleftrightarrow \frac{1}{(2+q)^t} \lg\left( \gamma \right) \ge -1 \Longleftrightarrow \lg\left( \frac{1}{\gamma} \right) \le (2+q)^t \\
        \Longleftrightarrow t \ge \frac{1}{\lg(2+q)}\lg\lg\left(\frac{1}{\gamma} \right).
    \end{multline*}
    Let us now look at the lower bound. As $c'(z_{t+1}) \ge z_{t+1}^p$ and $z_t \ge 0$ for all $t$,
    \begin{multline*}
        \frac{z_t}{(z_t + z_{t+1})^2} = \frac{c'(z_{t+1})}{4} \implies \frac{z_t}{(z_{t+1})^2} \ge \frac{z_{t+1}^p}{4} \implies z_{t+1} \le 4 z_t^{\frac{1}{2+p}}  \\
        \implies z_{t} \le 4 z_{t-1}^{\frac{1}{2+p}} \le 4^{1+\frac{1}{2+p}} z_{t-2}^{\frac{1}{(2+p)^2}} \le 4^{1+\frac{1}{2+p}+\frac{1}{(2+p)^2}} z_{t-3}^{\frac{1}{(2+p)^3}} \le \ldots \le 16 z_0^{\frac{1}{(2+p)^t}}.
    \end{multline*}
    Now, we want $z_t < \frac{1}{2}$. So, if we ensure that $16 z_0^{\frac{1}{(2+p)^t}} = 16 \gamma^{\frac{1}{(2+p)^t}} < \frac{1}{2}$, then we are done.
    \begin{multline*}
        16 \gamma^{\frac{1}{(2+p)^t}} < \frac{1}{2} \Longleftrightarrow \frac{1}{(2+p)^t} \lg(\gamma) < -5 \Longleftrightarrow \lg\left( \frac{1}{\gamma} \right) > 5 (2+p)^t \\
        \Longleftrightarrow t < \frac{1}{\lg(2+p)}\lg\lg\left(\frac{1}{\gamma} \right) - \lg(5).
    \end{multline*}
\qed\end{proof}

\begin{lemma}\label{lm:2hom:curved:2}
    Given $z_0 \ge \frac{1}{2}$, $z_t \ge 1 - \epsilon$ for all $t \ge \lg\lg(\frac{1}{\epsilon}) - \lg\lg(2+q) + O(1)$. On the other hand, given $z_0 \le \frac{1}{2}$, $z_t < 1 - \epsilon$ for all $t < \lg\lg(\frac{1}{\epsilon}) - \lg\lg(2+p) - O(1)$.
\end{lemma}
\begin{proof}
    As $z_{t+1} < 1$ for every $t$, which implies $c'(z_{t+1}) < 1$, we have
    \begin{align*}
        &\frac{z_t}{(z_t + z_{t+1})^2} = \frac{c'(z_{t+1})}{4} \implies 4 z_t \le (z_t + z_{t+1})^2 \implies z_{t+1} \ge \sqrt{z_t}(2 - \sqrt{z_t}).
    \end{align*}
    Let $\zeta_t = 1 - \sqrt{z_t} \Longleftrightarrow z_t = (1 - \zeta_t)^2$. From the above inequality, we have
    \begin{align}
        z_{t+1} \ge \sqrt{z_t}(2 &- \sqrt{z_t}) \Longleftrightarrow (1 - \zeta_{t+1})^2 \ge (1 - \zeta_t) (1 + \zeta_t) = 1 - \zeta_t^2 \nonumber \\
        &\Longleftrightarrow 1 + \zeta_{t+1}^2 - 2 \zeta_{t+1} \ge 1 - \zeta_t^2 \Longleftrightarrow 2\zeta_{t+1} \le \zeta_{t+1}^2 + \zeta_{t}^2. \label{eq:2hom:curved:1}
    \end{align}
    We will make use the inequality \eqref{eq:2hom:curved:1}, $2\zeta_{t+1} \le \zeta_{t+1}^2 + \zeta_{t}^2$, in our subsequent analysis. 
    
    We have been given that $z_{t}^p \le c'(z_{t}) \le z_{t}^q$ for all $t$. Making use of the upper bound, we get
    \begin{align*}
        &\frac{z_t}{(z_t + z_{t+1})^2} = \frac{c'(z_{t+1})}{4} \le \frac{z_{t+1}^q}{4} \implies \sqrt{z_t} \le z_{t+1}^{\frac{q}{2}} \left( \frac{z_t + z_{t+1}}{2} \right) \\
        &\Longleftrightarrow 1 - \zeta_t \le (1 - \zeta_{t+1})^{\frac{q}{2}} \left( \frac{(1 - \zeta_{t})^2 + (1 - \zeta_{t+1})^2}{2} \right) \\
        &\Longleftrightarrow 1 - \zeta_t \le (1 - \zeta_{t+1})^{\frac{q}{2}} \left( 1 - \zeta_{t} + \frac{\zeta_{t}^2 + \zeta_{t+1}^2 - 2\zeta_{t+1}}{2} \right) \\
        &\Longleftrightarrow 1 \le (1 - \zeta_{t+1})^{\frac{q}{2}} \left( 1 + \frac{\zeta_{t}^2 + \zeta_{t+1}^2 - 2\zeta_{t+1}}{2(1 - \zeta_t)} \right).
    \end{align*}
    From inequality \eqref{eq:2hom:curved:1}, $2\zeta_{t+1} \le \zeta_{t+1}^2 + \zeta_{t}^2 \implies \frac{\zeta_{t}^2 + \zeta_{t+1}^2 - 2\zeta_{t+1}}{2(1 - \zeta_t)} \le \zeta_{t}^2 + \zeta_{t+1}^2 - 2\zeta_{t+1}$ for $\zeta_t \le \frac{1}{2}$. Also, $(1 - \zeta_{t+1})^{\frac{q}{2}} \le e^{-\frac{\zeta_{t+1}q}{2}} \le 1 - \frac{\zeta_{t+1}q}{4}$ for $\zeta_{t+1} \le \frac{1}{8q}$. 
    (If $\zeta_{t+1} > \frac{1}{8q}$, then we can use an argument similar to Lemma~\ref{lm:2hom:curved:1} to reach $\zeta_{t+1} \le \frac{1}{8q}$ in a constant number of steps from $\zeta_0 \le \frac{1}{2}$.) 
    Putting things together, we get
    \begin{align*}
        &1 \le \left(1 - \frac{q\zeta_{t+1}}{4} \right) \left( 1 + \zeta_{t}^2 + \zeta_{t+1}^2 - 2\zeta_{t+1} \right) \\
        &\implies \zeta_{t+1} \left(2 + \frac{q}{4} \right) \le (\zeta_{t}^2 + \zeta_{t+1}^2) - \frac{q\zeta_{t+1}}{4} (\zeta_{t}^2 + \zeta_{t+1}^2 - 2\zeta_{t+1}) \\
        &\implies \zeta_{t+1} \le \frac{8}{8 + q} \zeta_{t}^2 \implies \zeta_{t} \le \left(\frac{8}{8 + q}\right)^{2^t-1} \zeta_{0}^{2^t}. 
    \end{align*}
    From the initial condition, we have $z_0 \ge \frac{1}{2} \implies \zeta_0 \le 1 - \frac{1}{\sqrt{2}} \le \frac{1}{2}$. We want $z_t \ge 1 - \epsilon \impliedby \sqrt{z_t} \ge \sqrt{1 - \epsilon}$. Using standard inequalities, $\sqrt{1 - \epsilon} \le \sqrt{e^{-\epsilon}} = e^{-\epsilon/2} \le 1 - \epsilon/4$ for every $\epsilon \le 1$. So, $\sqrt{z_t} \ge \sqrt{1 - \epsilon}$ is implied by $\sqrt{z_t} \ge 1 - \epsilon/4$. Putting everything together,
    \begin{multline*}
        z_t \ge 1 - \epsilon \impliedby \sqrt{z_t} \ge 1 - \frac{\epsilon}{4} \Longleftrightarrow \zeta_t \le \frac{\epsilon}{4} \impliedby \left(\frac{8}{8 + q}\right)^{2^t-1} \zeta_{0}^{2^t} \le \frac{\epsilon}{4} \\
        \impliedby \left(\frac{8 \zeta_{0}}{8 + q}\right)^{2^{t-1}} \le \frac{\epsilon}{4} \Longleftrightarrow 2^{t-1} \lg\left(\frac{8\zeta_0}{8+q}\right) \le \lg(\epsilon) - 2 \\
        \Longleftrightarrow 2^{t-1} \lg\left(\frac{1+q/8}{\zeta_0}\right) \ge 2 + \lg\left(\frac{1}{\epsilon}\right) 
        \impliedby t \ge \lg\lg\left(\frac{1}{\epsilon}\right) - \lg\lg\left(2+q\right) + O(1).
    \end{multline*}

    Let us now follow similar steps to derive the lower bound. Using $c'(z_{t}) \ge z_{t}^p$, we have
    \begin{align*}
        &\frac{z_t}{(z_t + z_{t+1})^2} = \frac{c'(z_{t+1})}{4} \ge \frac{z_{t+1}^p}{4} \implies \sqrt{z_t} \ge z_{t+1}^{\frac{p}{2}} \left( \frac{z_t + z_{t+1}}{2} \right) \\
        &\Longleftrightarrow 1 - \zeta_t \ge (1 - \zeta_{t+1})^{\frac{p}{2}} \left( \frac{(1 - \zeta_{t})^2 + (1 - \zeta_{t+1})^2}{2} \right) \\
        &\Longleftrightarrow 1 - \zeta_t \ge (1 - \zeta_{t+1})^{\frac{p}{2}} \left( 1 - \zeta_{t} + \frac{\zeta_{t}^2 + \zeta_{t+1}^2 - 2\zeta_{t+1}}{2} \right) \\
        &\Longleftrightarrow 1 \ge (1 - \zeta_{t+1})^{\frac{p}{2}} \left( 1 + \frac{\zeta_{t}^2 + \zeta_{t+1}^2 - 2\zeta_{t+1}}{2(1 - \zeta_t)} \right).
    \end{align*}
    From inequality \eqref{eq:2hom:curved:1}, $2\zeta_{t+1} \le \zeta_{t+1}^2 + \zeta_{t}^2 \implies \frac{\zeta_{t}^2 + \zeta_{t+1}^2 - 2\zeta_{t+1}}{2(1 - \zeta_t)} \ge \frac{\zeta_{t}^2 + \zeta_{t+1}^2 - 2\zeta_{t+1}}{2}$ for $\zeta_t \ge 0$. Also, $(1 - \zeta_{t+1})^{\frac{p}{2}} \ge e^{-\zeta_{t+1}p} \ge 1 -\zeta_{t+1}p$ for $\zeta_{t+1} \le \frac{1}{2}$. Putting things together, we get
    \begin{align*}
        &1 \ge \left(1 - p \zeta_{t+1} \right) \left( 1 + \frac{\zeta_{t}^2 + \zeta_{t+1}^2 - 2\zeta_{t+1}}{2} \right) \\
        &\implies \zeta_{t+1} \left(p + 1 \right) \ge \frac{\zeta_{t}^2 + \zeta_{t+1}^2}{2} - \frac{p\zeta_{t+1}}{2} (\zeta_{t}^2 + \zeta_{t+1}^2 - 2\zeta_{t+1}) \\
        &\implies \zeta_{t+1} \left(\frac{3p}{2} + 1 \right) \ge \frac{\zeta_{t}^2 + \zeta_{t+1}^2}{2} \implies \zeta_{t+1} \ge \frac{1}{2 + 3p} \zeta_{t}^2 \implies \zeta_{t} \ge \left(\frac{1}{2 + 3p}\right)^{2^t} \zeta_{0}^{2^t}.
    \end{align*}
    From the initial condition, we have $z_0 \le \frac{1}{2} \implies \zeta_0 \ge 1 - \frac{1}{\sqrt{2}} \ge \frac{1}{4}$. We want $z_t < 1 - \epsilon \impliedby \sqrt{z_t} < \sqrt{1 - \epsilon}$. As $\sqrt{1 - \epsilon} \ge 1 - \epsilon$ for every $\epsilon \le 1$, so $\sqrt{z_t} < \sqrt{1 - \epsilon}$ is implied by $\sqrt{z_t} < 1 - \epsilon$.
    \begin{multline*}
        z_t < 1 - \epsilon \impliedby \sqrt{z_t} < 1 - \epsilon \Longleftrightarrow \zeta_t > \epsilon \impliedby \left(\frac{1}{2 + 3p}\right)^{2^t} \zeta_{0}^{2^t} > \epsilon \\
        \Longleftrightarrow 2^{t} \lg\left(\frac{\zeta_0}{2+3p}\right) > \lg(\epsilon) \Longleftrightarrow 2^{t} \lg\left(\frac{2+3p}{4}\right) < \lg\left(\frac{1}{\epsilon}\right) \\
        \impliedby t < \lg\lg\left(\frac{1}{\epsilon}\right) - \lg\lg\left(2+p\right) - O(1).
    \end{multline*}
\qed\end{proof}
Lemmas~\ref{lm:2hom:curved:1}~and~\ref{lm:2hom:curved:2} give tight lower and upper bounds on the time required for $z_t$ to reach $1-\epsilon$ starting from $\gamma$. Applying Lemma~\ref{lm:lipschitz2} completes the proof for the upper bound. For the lower bound, notice that if $z_t \in [\frac{1}{2}, 1-\epsilon]$ and agent $i$ makes the move at time $t$, then $BR(x_{t,-i}) -  x_{t,i} = z_{t+1} -  z_{t-1} \ge z_{t+1} - z_t \ge \Omega(\epsilon)$ using arguments given in Lemma~\ref{lm:2hom:2} (if $z_t = 1 - \zeta \le 1 - \epsilon$, then $z_{t+1} \ge 1 - \zeta^2$, which implies $z_{t+1} - z_t \ge \zeta - \zeta^2 = \Omega(\zeta) = \Omega(\epsilon)$ for $\zeta \le 1/2$); applying Lemma~\ref{lm:revlipschitz} completes the proof.
\qed\end{proof}

%% file: proofs/dissum-randomB.tex
\begin{proof}[Lemma~\ref{lm:dissum:random}, remaining portion]
We now bound the rate of convergence. We will show that for every pair of consecutive time steps, the expected value of the potential decreases by at least a multiplicative factor of $(1 - \frac{\kappa (1-B) L^2}{\log(n)})$ for some constant $\kappa > 0$.

Let us look at two consecutive time steps $t$ and $t+1$. Let w.l.o.g. $f(\bz_t) > 0$ and $V_t \ge W_t$. So, $f(\bz_t) = V_t > 0$. Using Lemma~\ref{lm:partition} below, we know that that there exists $k, \ell \in [n]$ such that there are at least $k$ agents $i \in \cV_t$ with $v_{t,i} \ge \frac{V_t}{4 k \lg(n)}$ and at least $\ell$ agents $j \in \cW_t$ with $w_{t,j} \ge \frac{W_t}{4 \ell \lg(n)}$. 
\begin{lemma}\label{lm:partition}
    For any $(p_i)_{i \in [n]}$ with $p_i \ge 0$ for all $i \in [n]$ and $\sum_{i \in [n]} p_i = 1$, there exists $k \in [n]$ such that $ | \{ i \in [n] \mid p_i \ge \frac{1}{4 k \lg(n)} \} | \ge k$.
\end{lemma}
Say at time $t$, we pick one of the $k$ agents in $\cV_t$, say agent $i$, with $v_{t,i} \ge \frac{V_t}{4 k \lg(n)}$, and at time $t+1$, we pick one of the $\ell$ agents in $\cW_t$, say agent $j$, with $w_{t,j} \ge \frac{W_t}{4 \ell \lg(n)}$. We have the following possible scenarios:
\begin{enumerate}
    \item \label{lm:dissum:random:1} $W_t \le V_t - v_{t,i}$. Then agent $i$ will move from $\cV_t$ to $\cW_{t+1}$ after the transition and will have an updated value of $v_{t+1,i} = 0$ and $w_{t+1,i} \le B(V_t - v_{t,i} - W_t)$. So, the potential at time $t+1$ is
    \begin{align*}
        f(\bz_{t+1}) &\le \max(V_t - v_{t,i}, W_t + B(V_t - v_{t,i} - W_t)) \\
        &= \max(V_t - v_{t,i}, (1-B) W_t + B(V_t - v_{t,i})) \\
        &= V_t - v_{t,i} \le \left(1 - \frac{1}{4 k \lg(n)} \right) V_t = \left(1 - \frac{1}{4 k \lg(n)} \right) f(\bz_t).
    \end{align*}
    The probability that one of these $k$ agents with $v_{t,i} \ge \frac{V_t}{4 k \lg(n)}$ is picked is at least $kL$, so the expected value of the potential at time $t+1$ is $\Exp[f(\bz_{t+1}) | \bz_t] \le \left(1 - \frac{L}{4 \lg(n)} \right) f(\bz_t) \implies \Exp[f(\bz_{t+2}) | \bz_t] \le \Exp[f(\bz_{t+1}) | \bz_t] \le \left(1 - \frac{L}{4 \lg(n)} \right) f(\bz_t)$.


    \item \label{lm:dissum:random:2} $W_t > V_t - v_{t,i}$. As $W_t$ can be very close to $V_t$, we will need two transitions to guarantee sufficient progress. 
    As $v_{t,i} > V_t - W_t$, agent $i$ will stay in $\cV_{t+1}$ after the transition and will have an updated value of $v_{t+1,i} \le B (W_t - (V_t - v_{t,i}))$. We can bound $V_{t+1}$ as
    \begin{align*}
        V_{t+1} &\le V_t - v_{t,i} + B (W_t - (V_t - v_{t,i})) = (1-B) V_t + B W_t - (1-B)v_{t,i} \\
            &\le V_t - (1-B) v_{t,i} \le \left(1 - \frac{1-B}{4 k \lg(n)} \right) V_t.
    \end{align*}
    Now, at time $t+1$, agent $j$ with $w_{t+1,j} = w_{t,j} \ge \frac{W_t}{4 \ell \lg(n)} = \frac{W_{t+1}}{4 \ell \lg(n)}$ makes a move. We can have the following sub-cases depending upon the values of $w_{t+1,j}$, $W_{t+1}$, and $V_{t+1}$:
    \begin{enumerate}
        \item \label{lm:dissum:random:21} $V_{t+1} < W_{t+1} - w_{t+1,j}$. Then agent $j$ will move from $\cW_{t+1}$ to $\cV_{t+2}$ after the transition and will have an updated value of $w_{t+2,j} = 0$ and $v_{t+2,j} \le B(W_{t+1} - w_{t+1,j} - V_{t+1})$. So, the potential at time $t+2$ is
        \begin{align*}
            f(\bz_{t+2}) &\le \max(V_{t+1} + B(W_{t+1} - w_{t+1,j} - V_{t+1}), W_{t+1} - w_{t+1,j}) \\
            &= \max((1-B) V_{t+1} + B (W_{t+1} - w_{t+1,j}), W_{t+1} - w_{t+1,j}) \\
            &= W_{t+1} - w_{t+1,j} \le \left(1 - \frac{1}{4 \ell \lg(n)} \right) W_{t+1} \\
            &= \left(1 - \frac{1}{4 \ell \lg(n)} \right) f(\bz_{t+1}) \le \left(1 - \frac{1}{4 \ell \lg(n)} \right) f(\bz_t).
        \end{align*}

        \item \label{lm:dissum:random:22} $V_{t+1} \ge W_{t+1} - w_{t+1,j}$. Then agent $j$ will stay in $\cW_{t+2}$ after the transition and will have an updated value of $w_{t+2,j} \le B (V_{t+1} - (W_{t+1} - w_{t+1,j}))$. So, the potential at time $t+2$ is
        \begin{align*}
            f(\bz_{t+2}) &\le \max(V_{t+1}, W_{t+1} - w_{t+1,j} + B (V_{t+1} - (W_{t+1} - w_{t+1,j}))) \\
            &= \max(V_{t+1}, (1-B) (W_{t+1} - w_{t+1,j}) + B V_{t+1}) = V_{t+1} \\
            &\le \left(1 - \frac{1-B}{4 k \lg(n)} \right) V_t = \left(1 - \frac{1-B}{4 k \lg(n)} \right) f(\bz_t).
        \end{align*}        
    \end{enumerate}
    Now, combining the above two cases, we have 
    \begin{align*}
        f(\bz_{t+2}) &\le \max\left( \left(1 - \frac{1}{4 \ell \lg(n)} \right) f(\bz_t), \left(1 - \frac{1-B}{4 k \lg(n)} \right) f(\bz_t) \right) \\
        =& f(\bz_t) - f(\bz_t) \frac{1}{4 k \ell \lg(n)} \min( k, (1-B) \ell  ) \le \left( 1 - \frac{1-B}{4 k \ell \lg(n)} \right) f(\bz_t).
    \end{align*}
    Now, the probability that one of the $k$ agents with $v_{t,i} \ge \frac{V_t}{4 k \lg(n)}$ is picked at time $t$ is at least $kL$ and that one of the $\ell$ agents with $w_{t,i} \ge \frac{W_t}{4 \ell \lg(n)}$ is picked at time $t+1$ is at least $\ell L$, so the probability for the pair of transitions is at least $k \ell L^2$. So, the expected value of the potential is $\Exp[f(\bz_{t+2}) | \bz_t] \le \left(1 - \frac{(1-B) L^2}{4 \lg(n)} \right) f(\bz_t)$.
\end{enumerate}
\sloppy
Putting everything together, we have $\Exp[f(\bz_{t+2}) | \bz_t] \le \left(1 - \frac{(1-B) L^2}{4 \lg(n)} \right) f(\bz_t) \implies \Exp[f(\bz_{t+2})] = \Exp[\Exp[f(\bz_{t+2}) | \bz_t]] \le \left(1 - \frac{(1-B) L^2}{4 \lg(n)} \right) \Exp[ f(\bz_t) ]$. Therefore, after $T = \frac{8 \lg(n)}{(1-B) L^2} \ln\left(\frac{2 f(\bz_0)}{\epsilon\delta} \right)$ time steps, we have
\begin{align*}
    \Exp[f(\bz_T)] &\le \left(1 - \frac{(1-B) L^2}{4 \lg(n)} \right)^{T/2} f(\bz_0) \le e^{- \frac{(1-B) L^2}{4 \lg(n)} \frac{T}{2}} f(\bz_0) \\
    &= e^{-\ln\left(\frac{2 f(\bz_0)}{\epsilon\delta} \right)} f(\bz_0) \le \frac{\epsilon \delta}{2}.
\end{align*}
Finally, using Markov inequality we have $\Prob[f(\bz_{ T }) \ge \frac{\epsilon}{2}] \le \frac{2\Exp[f(\bz_{ T })]}{\epsilon} = \delta$. And it is easy to check that $|| \bz_T ||_1 \le 2 f(\bz_{ T }) < 2\frac{\epsilon}{2} = \epsilon$.
\qed\end{proof}

%% file: proofs/partition.tex
\begin{proof}[Lemma~\ref{lm:partition}]
We prove the theorem by contradiction. Let us assume that for every $k \in [n]$ there are strictly less than $k$ coordinates $i \in [n]$ such that $p_i \ge \frac{1}{4k\lg(n)}$. Then we have
\begin{align*}
    \sum_i p_i = 1 &= \sum_i \bone\left( \frac{1}{4 \lg(n)} \le v_i \le 1 \right) v_i \\
    &\qquad + \sum_i \sum_{j = 1}^{\lceil \lg(n) \rceil - 1} \bone\left( \frac{1}{2^{j+2} \lg(n)} \le v_i < \frac{1}{2^{j+1} \lg(n)} \right) v_i \\
    &\qquad + \sum_i \bone\left( 0 \le v_i < \frac{1}{2^{\lceil \lg(n) \rceil+1} \lg(n)} \right) v_i \\
    &\le \sum_i \bone\left( \frac{1}{4 \lg(n)} \le v_i \le 1 \right) \\
    &\qquad + \sum_i \sum_{j = 1}^{\lceil \lg(n) \rceil - 1} \bone\left( \frac{1}{2^{j+2} \lg(n)} \le v_i < \frac{1}{2^{j+1} \lg(n)} \right) \frac{1}{2^{j+1} \lg(n)} \\
    &\qquad + \sum_i \bone\left( 0 \le v_i < \frac{1}{2^{\lceil \lg(n) \rceil+1} \lg(n)} \right) \frac{1}{2^{\lceil \lg(n) \rceil+2} \lg(n)} \\
    &< \sum_{j = 1}^{\lceil \lg(n) \rceil - 1} \frac{2^j}{2^{j+1} \lg(n)} + \frac{n}{2^{\lceil \lg(n) \rceil+1} \lg(n)}, \\
    &\qquad \text{ as there are $ < k$ values $\ge \frac{1}{4 k \lg(n)}$ and at most $n$ values total,}  \\
    &\le \frac{\lceil \lg(n) \rceil}{2 \lg(n)} < 1.
\end{align*}
\qed\end{proof}

%% file: proofs/dissum-bestcase.tex
\begin{proof}[Lemma~\ref{lm:dissum:bestcase}]
The proof for the upper bound is a simplified version of the argument used in the proof of Lemma~\ref{lm:dissum:random}. As we do not have any randomness, we can deterministically pick the agent $i_t$ at time $t$. In particular, we pick the largest element in the larger side of $f(\bz_t)$ (where the two sides correspond to the sum of the positive and the negative values of $\bz_t$).

Formally, let $V_t = \sum_{j} \bone(z_{t,j} > 0) z_{t,j}$ and $W_t = \sum_{j} \bone(z_{t,j} < 0) z_{t,j}$. W.l.o.g. let $V_t \ge W_t$. So, $f(\bz_t) = V_t$. We pick $i_t = \argmax_{j \mid z_{t,j} > 0 } z_{t,j} \ge \frac{V_t}{n}$ and (if required) pick $i_{t+1} = \argmax_{j \mid z_{t,j} < 0 } (- z_{t,j}) \ge \frac{W_t}{n}$. We have the following cases based on the value of $z_{t,i_t}$, $V_t$, and $W_t$:
\begin{itemize}
    \item $W_t \le V_t - z_{t,i_t}$. Following the same steps as the proof of Lemma~\ref{lm:dissum:random} case \eqref{lm:dissum:random:1}, we have $f(\bz_{t+1}) \le V_t - z_{t,i_t} \le (1-\frac{1}{n}) V_t = (1-\frac{1}{n}) f(\bz_{t})$.
    \item $W_t > V_t - z_{t,i_t}$. Following the same steps as the proof of Lemma~\ref{lm:dissum:random} case \eqref{lm:dissum:random:2}, we have $V_{t+1} \le (1 - \frac{1-B}{n}) V_t$. Then either by following Lemma~\ref{lm:dissum:random} case \eqref{lm:dissum:random:21} we have $f(\bz_{t+2}) \le (1 - \frac{1}{n}) f(\bz_{t+1}) \le (1 - \frac{1}{n}) f(\bz_{t})$ or by following Lemma~\ref{lm:dissum:random} case \eqref{lm:dissum:random:22} we have $f(\bz_{t+2}) \le V_{t+1} \le (1 - \frac{1-B}{n}) f(\bz_{t})$.
\end{itemize}
Overall, we have the bound $f(\bz_{t+2}) \le (1 - \frac{1-B}{n}) f(\bz_{t})$. So, after $T = \frac{2n}{1-B} \ln(\frac{2f(\bz_0)}{\epsilon})$ steps, we have 
\begin{multline*}
    f(\bz_T) \le \left(1 - \frac{1-B}{n}\right)^{T/2} f(\bz_0) = e^{-\frac{(1-B)T}{2n}} f(\bz_0) \\
    = e^{-\ln\left(\frac{2f(\bz_0)}{\epsilon}\right)} f(\bz_0) = \frac{\epsilon}{2} \implies || \bz_1 || \le \epsilon.
\end{multline*}

Let us now prove the lower bound using two simple examples. 
In the first example, let $z_{0,i} = 1$ for all $i \in [n]$. $f(\bz_0) = n$. For $\epsilon < 1$, we will need every agent to play at least once to get $f(\bz_t) \le \epsilon < 1$. So, we need at least $n$ steps.

In the second example, let $z_{0,1} = z_{0,2} = \kappa$ for arbitrary $\kappa > 0$. $f(\bz_0) = \kappa$. Also let $\beta_{t,i}(\bz_t) = B \ge \frac{1}{2}$ for all $t$, $i$, and $\bz_t$. It can be easily checked that the fastest way to decrease the potential is to pick agents $1$ and $2$ alternately. And by doing that we have $f(\bz_{t+1}) = B f(\bz_t)$ for all $t \ge 1$, which implies $f(\bz_t) \ge B^{t-1} \kappa$, so
\begin{multline*}
    f(\bz_t) > \epsilon \impliedby B^{t-1} \kappa > \epsilon \Longleftrightarrow t-1 < \frac{1}{\ln(1/B)} \ln\left(\frac{\kappa}{\epsilon} \right) \\
    \impliedby t-1 < \frac{1}{2(1-B)} \ln\left(\frac{\kappa}{\epsilon} \right),
\end{multline*}
where the last inequality holds for $B \ge \frac{1}{2}$ because 
$2(1-B) = \ln(e^{2(1-B)}) \ge \ln(1 + 2(1-B)) \ge \ln(1 + \frac{1-B}{B}) = \ln(\frac{1}{B})$.
\qed\end{proof}

%% file: proofs/coupon.tex
\begin{proof}[Lemma~\ref{lm:coupon}]
Let $\tau$ be the time it takes to collect all $n$ coupons in the coupon collector problem where each coupon is selected w.p. $1/n$, the following result is well-known.
\begin{lemma}\cite{mitzenmacher2017probability}\label{lm:coupon:unif}
    For any constant $c > 0$, we have the following high probability bounds: (i) upper bound, $\Prob[\tau > n \ln n + cn] < e^{-c}$; (ii) lower bound, $\Prob[\tau < n \ln n - cn] < e^{-c}$.
\end{lemma}
The analysis underlying Lemma~\ref{lm:coupon:unif} goes as follows: The time to collect all coupons $\tau$ can be decomposed as $\tau = \sum_i \tau_i$, where $\tau_i$ is the time it takes to collect the $i$-th coupon. After $(i-1)$ coupons have been collected, the probability that we get a coupon that has not yet been collected in the next time step is equal to $p_i = \frac{n - i + 1}{n}$. So, $\tau_i$ corresponds to the time till the first head of a geometric random variable with parameter $p_i$.  In particular, $p_1 = 1$, $p_2 = (n-1)/n$, $p_3 = (n-2)/n$, and so on.

For the best-response dynamics, let $T = \sum_i T_i$ denote the time it takes for every agent to play at least once, where $T_i$ is the time between the $(i-1)$-th and the $i$-th unique agent. Remember that we are doing a worst-case analysis over the random selection process parameterized by $L$. In the first time step, we get the first unique agent w.p. $1$. In the second time step, the probability of selecting a new agent is at least $(n-1) L = \frac{(n-1) L n}{n} = L n p_2$. Similarly, we can show that for all $i$, after exactly $(i-1)$ agents have played at least once, the probability that we select an agent who has not yet played in the next time step is at least $L n p_i$. This implies that $\Prob[T_i \le k] \le Ln \Prob[\tau_i \le k]$ for all $i \in [n]$ and $k \ge 1$. Using the upper bound in Lemma~\ref{lm:coupon:unif}, we get $\Prob[T > \frac{1}{L}(\ln n + c)] < e^{-c}$, setting $c = \ln(1/\delta)$, we get $\Prob[T \le \frac{1}{L}(\ln n + \ln (\frac{1}{\delta}))] \ge 1 - \delta$ as required.

The idea for the lower bound is similar. As discussed earlier, after exactly $(i-1)$ unique agents have played, the total probability that the selection process can assign to the agents who have already played, in worst-case, is equal to $1 - (n-(i-1))L = 1 - nL p_i$. So, using the lower bound of Lemma~\ref{lm:coupon:unif}, we get $\Prob[T \ge \frac{1}{L} \log(n \delta)] \ge 1 - \delta$.
\qed\end{proof}

%% file: proofs/prop1.tex
\begin{proof}[Lemma~\ref{lm:prop1}]
    Let's first show that $s_t > 0$ for $t \ge 1$. For contradiction, say $s_t = 0$ for some $t \ge 1$. Let $i = i_{t-1}$ denote the agent that moved at time $t-1$. Notice that $s_t = 0 \implies s_{t-1,-i} = s_{t,-i} \le s_t = 0$. If $s_{t-1,-i} = 0$, then as a best response to this, agent $i$ must have played $x_{t,i} = a > 0$. So, we must have $s_t > 0$. Contradiction.

    Let us now prove that for any $t$, if $s_t < \frac{n^2}{(n-1) c'(0)}$, then $s_{t + 1} < \frac{n^2}{(n-1) c'(0)}$. If $c'(0) = 0$, e.g., when $c'(y) = y^r$ for some $r > 0$, then the inequality is satisfied trivially because $\frac{n^2}{(n-1) c'(0)} = \infty$. Assuming $c'(0) > 0$, we have the following two cases
    \begin{itemize}
        \item if $s_{t, -i_t} = 0$, then $s_{t + 1} = x_{t+1, i_{t}} = a < 1 <  \frac{n^2}{(n-1) c'(0)}$;
        \item if $0 < s_{t, -i_t} \le s_{t} < \frac{n^2}{(n-1) c'(0)}$, then from the first-order condition we have 
        \begin{multline*}
            \frac{s_{t, -i_t}}{(x_{t+1, i_{t}} + s_{t, -i_t})^2} - \frac{n-1}{n^2} c'(x_{t+1, i_{t}}) = 0 \\
            \implies \frac{s_{t, -i_t}}{(x_{t+1, i_{t}} + s_{t, -i_t})^2} = \frac{s_{t, -i_t}}{s_{t+1}^2} = \frac{n-1}{n^2} c'(x_{t+1, i_{t}}) \ge \frac{n-1}{n^2} c'(0) \\
            \implies s_{t+1} \le \sqrt{ s_{t, -i_t} \frac{n^2}{(n-1) c'(0)} } < \frac{n^2}{(n-1) c'(0)}.
        \end{multline*}
    \end{itemize}

    Let us now assume that at time $t$, $x_{t,i} > 0$, $x_{t,j} > 0$, and $s_t < \frac{n^2}{(n-1) c'(0)}$. Agent $i_t$ makes the move at time $t$. We know that $s_{t, -i_t} \ge \min( x_{t,i}, x_{t,j} ) > 0$ and $s_{t, -i_t} < s_{t} < \frac{n^2}{(n-1) c'(0)}$. Let us look at the first-order condition for agent $i_t$,
    \begin{align*}
        \frac{s_{t, -i_t}}{(x_{t+1, i_t} + s_{t, -i_t})^2} - \frac{n-1}{n^2} c'(x_{t+1, i_t}) = 0.
    \end{align*}
    If $c'(0) = 0$, then
    \[
        \frac{s_{t, -i_t}}{(0 + s_{t, -i_t})^2} - \frac{n-1}{n^2} c'(0) = \frac{1}{s_{t, -i_t}} > 0,
    \]
    so the BR must be strictly positive, i.e., $x_{t+1, i_t} > 0$, to satisfy the first-order condition (as the utility function is concave and its derivative decreasing). On the other hand, if $c'(0) > 0$, as $s_{t, -i_t} < \frac{n^2}{(n-1) c'(0)}$, we have
    \begin{align*}
        s_{t, -i_t} < \frac{n^2}{(n-1) c'(0)} \Longleftrightarrow  \frac{s_{t, -i_t}}{(0 + s_{t, -i_t})^2} - \frac{n-1}{n^2} c'(0) > 0.
    \end{align*}
    The same argument as before implies that $x_{t+1, i_t} > 0$ to ensure that the first-order condition is satisfied with equality.
\qed\end{proof}

%% file: proofs/part1.tex
\begin{proof}[Lemma~\ref{lm:part1}]
Let $T_1$ be a random variable that denotes the time it takes for every agent to make at least one move. We next prove that conditions \eqref{def:warmup:3} and \eqref{def:warmup:1} of the warm-up phase (Definition~\ref{def:warmup}) are satisfied for $t \ge T_1$.

Condition \eqref{def:warmup:3}.
Notice that Lemma~\ref{lm:prop1} \eqref{lm:prop1:1} implies $s_t > 0$ for $t \ge T_1 \ge n \ge 1$. For conciseness, let $\kappa = \frac{n^2}{(n-1)c'(0)}$.
Let us now prove that $s_t < \kappa$ for $t \ge T_1$. Note that it is enough to prove $s_{T_1} < \kappa$ because $s_{T_1} < \kappa$ implies $s_t < \kappa$ for all $t \ge T_1$ using Lemma~\ref{lm:prop1} \eqref{lm:prop1:2}.
For contradiction, let us assume that $s_{T_1} \ge \kappa$. Let $t = T_1 - 1$. 
\begin{itemize}
    \item If $s_{t, -i_t} = 0$, then $s_{T_1} = s_{t+1} = x_{t+1,i_t} = a < \kappa$, which contradicts $s_{T_1} \ge \kappa$.
    \item If $0 < s_{t, -i_t} < \kappa$, then applying the first-order condition for agent $i_t$, we have 
    $ \frac{s_{t, -i_t}}{s_{t+1}^2} - \frac{n-1}{n^2} c'(x_{t+1, i_t}) = 0 \implies s_{t+1} < \sqrt{\kappa \frac{n^2}{(n-1)c'(0)}} = \kappa$. So, $s_{T_1} = s_{t+1}  < \kappa$, which also leads to contradiction. 
    \item If $s_{t, -i_t} \ge \kappa$, then $x_{t+1,i_t} = 0$ and $s_t \ge s_{t, -i_t} \ge \kappa$. Doing induction in the backward direction, starting from $\tau = t = T_1 - 1$ and going to $\tau = 0$, $x_{\tau+1,i_\tau} = 0$ for every $\tau < T_1$. As every agent $i$ makes at least one move before $T_1$, so every agent $i \in [n]$ has $x_{T_1, i} = 0$. So, $s_{T_1} = 0$. Contradiction. 
\end{itemize}
This completes the proof for condition \eqref{def:warmup:3} of the warm-up phase.

Condition \eqref{def:warmup:1}.
Let us now prove that $x_{t,i} \le \frac{n^2}{4(n-1)}$ for all $t \ge T_1$ and $i \in [n]$. Fix a $t \ge T_1$ and an $i \in [n]$. As every agent has played at least once before $T_1$, so $i$ must have made a move before $t$. Let $\tau < t$ denote the most recent time before $t$ when agent $i$ made the best-response move. 
By definition of $\tau$, $i$ did not play between $\tau+1$ and $t$, so $x_{t,i} = x_{\tau+1,i}$. Finally, $x_{\tau+1,i} \le \frac{n^2}{4(n-1)}$ because:
\begin{itemize}
    \item If $s_{\tau,-i} = 0$, then $x_{\tau+1,i} = a \le \frac{n^2}{4(n-1)}$.
    \item If $0 < s_{\tau,-i} < \frac{n^2}{(n-1) c'(0)
    }$, then we know that $0 < x_{\tau+1,i} \le s_{\tau+1} < \frac{n^2}{(n-1) c'(0)}$ from Lemma~\ref{lm:prop1}. If $x_{\tau+1,i} \le 1$, then $x_{\tau+1,i} \le \frac{n^2}{4(n-1)}$ because $n \ge 3$. On the other hand, if $x_{\tau+1,i} \ge 1$, then $c'(x_{\tau+1,i}) \ge 1$, and using the first-order optimality condition we have
    \begin{align*}
        &\frac{s_{\tau,-i}}{(x_{\tau+1, i} + s_{\tau,-i})^2} = \frac{n-1}{n^2} c'(x_{\tau+1, i}) \ge \frac{n-1}{n^2} c'(1) = \frac{n-1}{n^2} \\
        &\implies x_{\tau+1, i} \le \frac{n}{\sqrt{n-1}} \sqrt{s_{\tau,-i}} - s_{\tau,-i} \le \max_{z \ge 0} \left( \frac{n}{\sqrt{n-1}} z - z^2 \right).
    \end{align*}
    As $\frac{n}{\sqrt{n-1}} z - z^2$ is maximized at $z = \frac{n}{2\sqrt{n-1}}$, so $x_{\tau+1, i} \le \frac{n^2}{4(n-1)}$ as required.
\end{itemize}
This completes the proof for condition \eqref{def:warmup:1} of the warm-up phase.

We have shown that $s_t > 0$ for $t \ge T_1$, so there must be at least one agent with positive output for all $t \ge T_1$. Let $i$ be an agent that has positive output $x_{T_1,i} > 0$ at time $T_1$. We also know from condition \eqref{def:warmup:3}, which we proved earlier, that $x_{T_1,i} \le s_{T_1} < \frac{n^2}{(n-1) c'(0)}$. Although $x_{T_1,i} > 0$, it is possible that $i$ is the only agent with positive output and every other agent $j \neq i$ has $x_{T_1,j} = 0$. We next resolve this scenario.

Condition \eqref{def:warmup:2}.
Let $T_2$ denote the additional steps after $T_1$, if any, required to get at least two agents with positive output. Notice that $T_2$ can be upper bounded by the time it takes to get the first head (select an agent $j \neq i$) of a geometric random variable with parameter $p \ge (n-1)L$ because each $j \neq i$ is assigned a probability of at least $L$ at each time step. When a $j \neq i$ is selected at time $t = T_2 - 1$, then $x_{T_2,j} = x_{t+1,j} \in (0, \frac{n^2}{(n-1)c'(0)})$, as required, using Lemma~\ref{lm:prop1} \eqref{lm:prop1:3}. Further, Lemma~\ref{lm:prop1} \eqref{lm:prop1:3} also implies that there will always be at least two agents with positive output this time onward. So, for $t \ge T_1 + T_2$, the action profile $x_t$ satisfies condition \eqref{def:warmup:2} of the warm-up phase.

To summarize, for $t \ge T_1 + T_2$, the action profile $x_t$ satisfies all conditions required for the completion of the warm-up phase (Definition~\ref{def:warmup}). So, $T_{warm} \le T_1 + T_2$. Let us now prove a high probability upper bound on $T_1 + T_2$, which holds for $T_{warm}$ as well. 
From Lemma~\ref{lm:coupon}, we know that $\Prob[T_1 > \frac{1}{L}\ln (\frac{n}{\delta'})] < \delta'$ for any $\delta' \in (0,1)$. Set $\delta' = \frac{\delta}{2}$, we get $\Prob[T_1 > \frac{1}{L}\ln (\frac{2n}{\delta})] < \frac{\delta}{2}$.
As $T_2$ underestimates the time till the first head for a geometric random variable with parameter $(n-1)L$, we have $\Prob[T_2 > k] \le (1 - (n-1)L))^k < e^{-(n-1)Lk}$. Setting $k = \frac{1}{(n-1)L}\ln(\frac{2}{\delta})$, we get $\Prob[T_2 > \frac{1}{(n-1)L}\ln(\frac{2}{\delta})] < \frac{\delta}{2}$.
Using union bound, we get $T_{warm} \le T_1 + T_2 \le \frac{1}{L}\ln (\frac{2n}{\delta}) +\frac{1}{(n-1)L}\ln(\frac{2}{\delta}) = O(\frac{1}{L}\log(\frac{n}{\delta}))$ w.p. $1- \delta$ as required.
\qed\end{proof}

%% file: proofs/good_domain1.tex
\begin{proof}[Lemma~\ref{lm:good_domain1}]
    Let $i = i_t$. Notice that $z_{t+1, i} = 0 \Longleftrightarrow x_{t+1, i} = 1$. If $x_{t+1, i} = 1$, the first order condition for agent $i$ is
    \begin{align*}
        &\frac{s_{t, -i}}{(x_{t+1, i} + s_{t, -i})^2} = \frac{n-1}{n^2} c'(x_{t+1,i}) \Longleftrightarrow \frac{s_{t, -i}}{(1 + s_{t, -i})^2} = \frac{n-1}{n^2} \\
        &\Longleftrightarrow (n-1) (1 + s_{t, -i})^2 - n^2 s_{t, -i} = 0.
    \end{align*}
    The two solutions for the quadratic equation above are
    \begin{align*}
        s_{t, -i} &= n-1 \implies (n-1) (1 + s_{t, -i})^2 - n^2 s_{t, -i} \\
            &= (n-1) (1 + n - 1)^2 - n^2 (n-1) = 0, \\
        s_{t, -i} &= \frac{1}{n-1} \implies (n-1) (1 + s_{t, -i})^2 - n^2 s_{t, -i} \\
            &= (n-1) \left(1 + \frac{1}{n-1} \right)^2 + n^2 \frac{1}{n-1} = (n-1) \left(\frac{n}{n-1} \right)^2 + n^2 \frac{1}{n-1} = 0.
    \end{align*}
    Further, when $s_{t, -i} \in (\frac{1}{n-1}, n-1)$, then $\frac{s_{t, -i}}{(1 + s_{t, -i})^2} > \frac{n-1}{n^2}$, which implies that the BR to $s_{t, -i}$ is strictly more than $1$. Similarly, we can verify that if $s_{t, -i} \in (0, \frac{1}{n-1})$ or $s_{t, -i} > n-1$, then $x_{t+1, i} < 1$. Writing this in terms of $z$ and $\sigma$, we have
    \begin{itemize}
        \item $s_{t, -i} \in [\frac{1}{n-1}, n-1] \Longleftrightarrow \sigma_{t, -i} \in [\frac{1}{n-1} - (n-1), 0]$, then $x_{t+1, i} \ge 1 \Longleftrightarrow z_{t+1, i} \ge 0$, and
        \item $s_{t, -i} \ge n-1 \Longleftrightarrow \sigma_{t, -i} \ge 0$, then $x_{t+1, i} \le 1 \Longleftrightarrow z_{t+1, i} \le 0$.
    \end{itemize}
    So, we have shown that if $s_{t, -i} \ge \frac{1}{n-1}$, then $z_{t+1, i}$ has the opposite sign as $\sigma_{t, -i}$, as required. We now upper bound the ratio $\frac{-z_{t+1, i}}{\sigma_{t, -i}}$. 

    Let $r = \frac{-z_{t+1, i}}{\sigma_{t, -i}}$, we want to upper bound $r$. Let us write the first order condition for agent $i$ using $z$ and $\sigma$
    \begin{align*}
        \frac{s_{t, -i}}{(x_{t+1, i} + s_{t, -i})^2} = \frac{n-1}{n^2} c'(x_{t+1,i}) \Longleftrightarrow \frac{n-1 + \sigma_{t, -i}}{(n + z_{t+1,i} + \sigma_{t, -i})^2} = \frac{n-1}{n^2} c'(x_{t+1,i}).
    \end{align*}
    If $x_{t+1,i} > 1 \Longleftrightarrow z_{t+1,i} > 0$, then $c'(x_{t+1,i}) \ge c'(1) = 1$. Also, $z_{t+1,i} > 0 \implies \sigma_{t, -i} < 0$. Using these, we get
    \begin{align*}
        &\frac{n-1 + \sigma_{t, -i}}{(n + z_{t+1,i} + \sigma_{t, -i})^2} = \frac{n-1 + \sigma_{t, -i}}{(n + (1-r)\sigma_{t, -i})^2} \ge \frac{n-1}{n^2} \\
        &\implies 1 + \frac{\sigma_{t, -i}}{n-1} \ge \left( 1 + \frac{(1-r)\sigma_{t, -i}}{n} \right)^2 \implies \frac{(1-r)\sigma_{t, -i}}{n} \le \sqrt{1 + \frac{\sigma_{t, -i}}{n-1}} - 1 \\
        &\implies 1 - r \ge \frac{n}{\sigma_{t, -i}} \left( \sqrt{1 + \frac{\sigma_{t, -i}}{n-1}} - 1 \right) \text{ as $\sigma_{t, -i} < 0$}.
    \end{align*}
    Let $g(y) = \frac{\sqrt{1 + y} - 1}{y}$. Notice that $1 - r \ge \frac{n}{n-1} g(\frac{\sigma_{t, -i}}{n-1})$. As $\sigma_{t, -i} < 0$, so $ \frac{\sigma_{t, -i}}{n-1} < 0$. Further, as $\sigma_{t, -i} \ge -(n-1) + \frac{1}{n-1}$, so $\frac{\sigma_{t, -i}}{n-1} > -1$. We will later lower bound $g(y)$ in $y \in (-1, 0]$, but before that let us look at the case $x_{t+1,i} < 1$. 

    If $x_{t+1,i} < 1 \Longleftrightarrow z_{t+1,i} < 0$, then $c'(x_{t+1,i}) \le c'(1) = 1$. Also, as $z_{t+1,i} < 0 \implies \sigma_{t, -i} > 0$. Using these, we get
    \begin{align*}
        &\frac{n-1 + \sigma_{t, -i}}{(n + (1-r)\sigma_{t, -i})^2} \le \frac{n-1}{n^2} \implies 1 + \frac{\sigma_{t, -i}}{n-1} \le \left( 1 + \frac{(1-r)\sigma_{t, -i}}{n} \right)^2 \\
        &\implies \frac{(1-r)\sigma_{t, -i}}{n} \ge \sqrt{1 + \frac{\sigma_{t, -i}}{n-1}} - 1 \implies 1 - r \ge \frac{n}{\sigma_{t, -i}} \left( \sqrt{1 + \frac{\sigma_{t, -i}}{n-1}} - 1 \right),
    \end{align*}
    where the last inequality holds because $\sigma_{t, -i} > 0$. Again, for $g(y) = \frac{\sqrt{1 + y} - 1}{y}$, we have the same inequality as before $1 - r \ge \frac{n}{n-1} g(\frac{\sigma_{t, -i}}{n-1})$. Now, if $\sigma_{t, -i} \ge 2$, we trivially have $r = \frac{-z_{t+1, i}}{\sigma_{t, -i}} \le \frac{1}{2}$ because $z_{t+1, i} \ge -1$ always. So, we can assume $\sigma_{t, -i} \le 2 \implies \frac{\sigma_{t, -i}}{n-1} \le \frac{2}{n-1}$. Earlier, for $z_{t+1,i} > 0$, we needed to lower bound $g(y)$ in domain $y \in [-1, 0]$, and now for $z_{t+1,i} < 0$, we need to lower bound $g(y)$ in domain $[0, \frac{2}{n-1}]$.

    Let us differentiate $g(y) = \frac{\sqrt{1 + y} - 1}{y}$ for $y > -1$, we get
    \begin{align*}
        g&'(y) = \frac{1}{2 y \sqrt{1+y}} - \frac{\sqrt{1 + y} - 1}{y^2} = \frac{y - 2 \sqrt{1+y} (\sqrt{1+y} - 1)}{2 y^2 \sqrt{1+y}} \\
        &= \frac{y - 2 (1+y) + 2 \sqrt{1+y}}{2 y^2 \sqrt{1+y}} = \frac{-1 - (1+y) + 2 \sqrt{1+y}}{2 y^2 \sqrt{1+y}} = \frac{- (1 -\sqrt{1+y})^2}{2 y^2 \sqrt{1+y}} < 0.
    \end{align*}
    So, $g(y)$ is a decreasing function for all $y > -1$. Therefore, the minimum value of $g(y)$ in domain $[0, \frac{2}{n-1}]$ occurs at $y = \frac{2}{n-1}$. Plugging it in, we get
    \begin{align*}
        1-r \ge \frac{n}{n-1} g\left(\frac{\sigma_{t, -i}}{n-1}\right) \ge \frac{n}{n-1} g\left(\frac{2}{n-1}\right) = \frac{n}{2} \left( \sqrt{1 + \frac{2}{n-1}} - 1 \right).
    \end{align*}
    Let us now lower bound $\frac{n}{2} \left( \sqrt{1 + \frac{2}{n-1}} - 1 \right)$. Let $v = 1 + \frac{2}{n-1} \Longleftrightarrow n = 1 + \frac{2}{v-1} = \frac{v+1}{v-1}$. Notice that $v \ge 1$ as $n \ge 2$, and as $n \rightarrow \infty$, $v \rightarrow 1$. Let $h(v) = \frac{n}{2} \left( \sqrt{1 + \frac{2}{n-1}} - 1 \right) = \frac{v+1}{2(v-1)}(\sqrt{v}-1) = \frac{v+1}{2(\sqrt{v} + 1)}$. We need to lower bound $h(v)$ for $v > 1$. Differentiating $h(v)$ w.r.t. $v$, we have
    \begin{multline*}
        h'(v) = \frac{1}{2 (\sqrt{v}+1)} - \frac{v+1}{2} \frac{1}{(\sqrt{v} + 1)^2} \frac{1}{2 \sqrt{v}} \\
        = \frac{2\sqrt{v}(\sqrt{v}+1) - (v+1)}{4 \sqrt{v} (\sqrt{v} + 1)^2} = \frac{v + 2\sqrt{v} - 1}{4 \sqrt{v} (\sqrt{v} + 1)^2} > 0, \text{ as $v > 1$.}
    \end{multline*}
    So, the minimum value of $h(v)$ in domain $v \in (1, \infty)$ occurs at $v = 1$, and we have $h(v) \ge h(1) = \frac{1}{2}$. So, we have proven that $1-r \ge \frac{1}{2} \implies r \le \frac{1}{2}$.
\qed\end{proof}

%% file: proofs/good_domain2.tex
\begin{proof}[Lemma~\ref{lm:good_domain2}]
Note that we are assuming the completion of the warm-up phase.
Let $\bx_t$ have at least two agents $i$ and $j \neq i$ with $x_{t,i} \ge \frac{1}{n-1}$ and $x_{t,j} \ge \frac{1}{n-1}$. We shall prove that $\bx_{t+1}$ also has two agents with output at least $\frac{1}{n-1}$. By induction, this implies the lemma.

If an agent $i_t \notin \{i,j\}$ makes the move at time $t$. Then $x_{t+1,i} = x_{t,i} \ge \frac{1}{n-1}$ and $x_{t+1,j} = x_{t,j} \ge \frac{1}{n-1}$, and we are done trivially. So, let us assume that one of $i$ or $j$, w.l.o.g. say $i$, makes the move at time $t$. If there is an agent $k \notin \{i,j\}$ with $x_{t,i} \ge \frac{1}{n-1}$, then again we are done trivially because $j$ and $k$ will have output at least $\frac{1}{n-1}$ at time $t+1$. So, let us assume that $x_{t,k} < \frac{1}{n-1}$ for all $k \notin \{i,j\}$.

Now, as the warm-up phase has completed, so $x_{t,j} \le \frac{n^2}{4(n-1)}$. Therefore, $s_{t,-i} = x_{t,j} + \sum_{k \notin \{i,j\}} x_{t,k} \le \frac{n^2}{4(n-1)} + \frac{n-2}{n-1} \le \frac{n^2}{2(n-1)}$. Also, $s_{t,-i} \ge x_{t,j} \ge \frac{1}{n-1}$. Let us now look at the first order condition for agent $i$. If $x_{t+1,i} \ge 1$, then we are done trivially again. So, assume $x_{t+1,i} \le 1$, which implies $c'(x_{t+1,i}) \le c'(1) = 1$. We have
\begin{align*}
    \frac{s_{t,-i}}{(x_{t+1,i} + s_{t,-i})^2} = \frac{n-1}{n^2} c'(x_{t+1,i}) \le \frac{n-1}{n^2} \implies x_{t+1,i} \ge \sqrt{\frac{n^2 s_{t,-i}}{n-1}} - s_{t,-i}.
\end{align*}
Notice that the function $g(y) = \sqrt{\frac{n^2 y}{n-1}} - y$ is concave in $y$, so the minimum values occur at extreme points. As $s_{t,-i} \in [\frac{1}{n-1}, \frac{n^2}{2(n-1)}]$, plugging in the extreme points we get:
\begin{align*}
    g\left(\frac{1}{n-1}\right) &\ge \sqrt{\frac{n^2}{(n-1)^2}} - \frac{1}{n-1} = \frac{n}{n-1} - \frac{1}{n-1} = 1 \ge \frac{1}{n-1}, \\
    g\left(\frac{n^2}{2(n-1)}\right) &\ge \sqrt{\frac{n^2 n^2}{(n-1) 2 (n-1)}} - \frac{n^2}{2(n-1)} = \frac{n^2}{n-1} \frac{\sqrt{2} - 1}{2} \ge \frac{1}{n-1},
\end{align*}
where the last inequality holds for all $n \ge 3$. So, we get $x_{t+1,i} \ge \frac{1}{n-1}$ as required.
\qed\end{proof}

%% file: proofs/good_domain3.tex
\begin{proof}[Lemma~\ref{lm:good_domain3}]
Let us look at the sequence of moves after the warm-up phase. Let $t \ge T_{warm}$ and $i = i_t$, and let's assume that there are less than two agents with output at least $\frac{1}{n-1}$ at time $t$. Then, we have the following cases based on the value of $s_{t, -i}$.

Let $s_{t, -i} \ge \frac{1}{n-1}$. Then following exactly the same argument as the proof of Lemma~\ref{lm:good_domain2}, we have $x_{t+1,i} \ge \frac{1}{n-1}$. Now, at time $t+1$, if we have two agents with output $\ge \frac{1}{n-1}$, we are done. Else, let $j \neq i$ be the agent who plays the next BR move, say at time $\tau \ge t+1$ (ignoring the redundant consecutive BR moves by $i$ that do not change the output profile). Again, following exactly the same argument as the proof of Lemma~\ref{lm:good_domain2}, we have $x_{\tau+1,j} \ge \frac{1}{n-1}$. Finally, the time taken to get this non-redundant move by an agent $j \neq i$ is at most $m = \frac{1}{(n-1)L} \ln(\frac{1}{\delta})$ w.p. $1-\delta$ because the probability that agent $i$ makes $m$ consecutive moves is $\le (1 - (n-1)L)^{m} \le e^{-(n-1)Lm} \le \delta$.

Let $s_{t,-i} < \frac{1}{n-1}$. Then the two possible scenarios are either $x_{t+1,i} \ge \frac{1}{n-1}$ or $x_{t+1,i} < \frac{1}{n-1} \le 1$. If $x_{t+1,i} \le 1$, then $c'(x_{t+1,i}) \le c'(1) = 1$, and the first-order condition implies
\begin{align*}
    &\frac{s_{t,-i}}{(x_{t+1,i} + s_{t,-i})^2} = \frac{n-1}{n^2} c'(x_{t+1,i}) \le \frac{n-1}{n^2} \implies x_{t+1,i} \ge \sqrt{\frac{n^2 s_{t,-i}}{n-1}} - s_{t,-i} \\
    &\implies x_{t+1,i} \ge \sqrt{s_{t,-i}} \left( \frac{n}{\sqrt{n-1}} - \sqrt{s_{t,-i}} \right) \ge \sqrt{s_{t,-i}} \left( \frac{n}{\sqrt{n-1}} - \frac{1}{\sqrt{n-1}} \right)\\
    &\implies x_{t+1,i} \ge \sqrt{(n-1) s_{t,-i}}.
\end{align*}
Now, let $j \neq i$ be the agent who plays the next BR move, say at time $\tau$ (ignoring the redundant consecutive BR moves by $i$ for now). At $\tau$, $s_{\tau, -j} = s_{t+1, -j} \ge x_{t+1,i} \ge \sqrt{(n-1) s_{t,-i}}$.

Let $(y_t)_{t \ge 0}$ be a sequence where $y_0 = \gamma = \min_{j \in [n]} s_{T_{warm}, -j} < \frac{1}{n-1}$ and $y_{t+1} \ge \sqrt{(n-1) y_t}$. Notice that $y_t$ tracks the evolution of $s_{t, -i_t}$ for $t \ge T_{warm}$ assuming there are no redundant BR moves. By the definition of $T_{warm}$, there are at least two agents with strictly positive output, so $\gamma > 0$. We want to measure the time it takes for $y_t$ to reach $\frac{1}{n-1}$.
\begin{align*}
    y_{t} \ge \sqrt{(n-1) y_{t-1}} \ge (n-1)^{\frac{1}{2} + \frac{1}{4}} y_{t-2}^{\frac{1}{4}} \ge \ldots \ge (n-1)^{1-\frac{1}{2^{t}}} y_{0}^{\frac{1}{2^{t}}} = (n-1)^{1-\frac{1}{2^{t}}} \gamma^{\frac{1}{2^{t}}}.
\end{align*}
We want to ensure $y_t \ge \frac{1}{n-1}$, which is implied by
\begin{align*}
    (n-1)^{1-\frac{1}{2^{t}}} \gamma^{\frac{1}{2^{t}}} \ge \frac{1}{n-1} \impliedby \gamma^{\frac{1}{2^{t}}} \ge \frac{1}{2} \impliedby \frac{1}{2^t} \lg(\gamma) \ge -1 \impliedby t \ge \lg\lg(\frac{1}{\gamma}).
\end{align*}
So, we need $\kappa = O(\log\log(\frac{1}{\gamma}))$ non-redundant moves. Now, let's provide a high probability bound on the time it takes to have $\kappa$ non-redundant moves. At each time step, we make a non-redundant move with probability at least $(n-1)L$. So, we need to find the time it takes to get $\kappa$ heads of a geometric random variable when the probability of getting a head is $p \ge (n-1)L$. 

Let $m = \frac{2}{p}\max(\kappa, \frac{1}{p} \ln(\frac{1}{\delta}))$. Let $Z_i$ be a random variable that takes value $-1$ (head) w.p. $p$ and $0$ (tail) w.p. $1-p$. $\Exp[\sum_{i=1}^m Z_i] = -pm$. Less than $\kappa$ heads ($-1$ values) corresponds to having $\sum_{i=1}^m Z_i > -\kappa$. Using Hoeffding's inequality, we have
\begin{align*}
    \Exp[\sum_{i=1}^m Z_i > -\kappa] &= \Exp[\sum_{i=1}^m Z_i - \Exp[\sum_{i=1}^m Z_i] > -\kappa + pm] \\
    &\le e^{\frac{-2 (pm - \kappa)^2}{m}} \le e^{\frac{-2 (pm - \frac{pm}{2})^2}{m}} \le e^{\frac{-m p^2}{2}} \le \delta.
\end{align*}
So, after $\frac{2}{(n-1)L}\max(\kappa, \frac{1}{(n-1)L} \ln(\frac{1}{\delta}))$ steps, the probability of getting less than $\kappa$ non-redundant moves is bounded above by $\delta$, as required.
\qed\end{proof}

%% file: proofs/outputlb.tex
\begin{proof}[Lemma~\ref{lm:outputlb}]
If $T_{warm} = 0$, i.e., all required conditions for completion of the warm-up phase are satisfied by the initial state $x_0$, then we trivially get $\gamma = \min_{j} s_{0,-j} \ge \min(\cA) \ge \gamma_{lb}$.

Let us assume that $T_{warm} > 0$. Let us focus on the time step $T_{warm} - 1$. Let $t = T_{warm} - 1 \Longleftrightarrow T_{warm} = t+1$.

As $t+1$ is the smallest time when all conditions for completion of the warm-up phase, Definition~\ref{def:warmup}, are satisfied, therefore at time $t$, there must be at least one violation. We do a case analysis depending upon which condition was violated at $t$.  Let $\kappa = \frac{n^2}{(n-1) c'(0)}$ for conciseness. Let $i = i_t$ be the agent who makes the transition at time $t$.

\paragraph{\textbf{Case 1:}} Definition~\ref{def:warmup} condition \eqref{def:warmup:2} violated at $t$, i.e., there is only one agent $j$ with $x_{t,j} > 0$. 

First, notice that at time $t$ an agent $i \neq j$ makes a transition to a positive $x_{t+1,i}$ to satisfy all three conditions of Definition~\ref{def:warmup}. Check that the conditions \eqref{def:warmup:1} and \eqref{def:warmup:3} of Definition~\ref{def:warmup} must not have been violated at time $t$ because, then, in a single step, we could not have satisfied all three conditions. This implies that $s_t = x_{t,j} \le \frac{n^2}{4(n-1)} < \frac{n^2}{(n-1) c'(0)} = \kappa$.

Let us now trace our steps back from $t$ to $0$. We claim that all transitions before time $t$ were made by agent $j$, i.e., for every $\tau < t$, $j = i_{\tau}$. If not, let $\tau < t$ be the most recent transition by an agent $k \neq j$. As $x_{\tau+1,k} = x_{t,k} = 0$, therefore $s_{\tau,-k} \ge \kappa$, but $s_{\tau,-k} = s_{t,-k} = x_{t,j} < \kappa$. Contradiction. 

As $j$ makes all the transitions before time $t$, we have either (i) if $t = 0$, then $x_{t,j} = x_{0,j}$; or (ii) if $t > 0$, then $x_{t,j} = a$ as a response to $0$ output by everyone else. Notice that $\gamma \ge \min(x_{t+1,i}, x_{t+1,j})$. Further, we have either $x_{t+1,i} \ge 1$ or $x_{t+1,i} \le 1 \implies c'(x_{t+1,i}) \le 1$ and from the first-order condition
\begin{multline*}
    x_{t+1,i} \ge \sqrt{s_{t,-i}} \left( \frac{n}{\sqrt{n-1}} - \sqrt{s_{t,-i}} \right) = \sqrt{x_{t,j}} \left( \frac{n}{\sqrt{n-1}} - \sqrt{x_{t,j}} \right) \\
    \ge \sqrt{x_{t,j}} \left( \frac{n}{\sqrt{n-1}} - \frac{n}{2\sqrt{n-1}} \right) \ge \sqrt{x_{t,j}} \ge \min(1, x_{t,j}).
\end{multline*}
So, $\gamma \ge \min(\{a\} \cup \cA)$.

\paragraph{\textbf{Case 2:}} Definition~\ref{def:warmup} condition \eqref{def:warmup:3} violated at $t$, i.e., total output $s_t \ge \kappa$, but condition \eqref{def:warmup:2} is satisfied.

Let $i$ be the agent that makes the move at time $t$ to decrease the total output from $s_t \ge \kappa$ to $s_t < \kappa$. We argued in Case 1 that condition \eqref{def:warmup:2} must have been satisfied at $t$. So, there are at least two agents with strictly positive output at $t$. Let $j \neq i$ be an agent other than $i$ that has $x_{t,j} > 0$. We claim that $x_{t,j} = x_{0,j}$.

We prove our claim by contradiction. Let us trace our steps back from $t$. Let agent $k$ make the transition at $t-1$. We show that $k \notin \{i, j\}$. Notice that $x_{t,k} = 0$ because if $x_{t,k} > 0$ then either: (i) $s_{t-1,-k} = 0$, but this is not possible as $s_{t-1,-k} \ge \min(x_{t,i}, x_{t,j}) > 0$; (ii) $0 < s_{t-1, -k} < \kappa$, but this is also not possible because then $s_t < \kappa$ by Lemma~\ref{lm:prop1}, but $s_t \ge \kappa$. As we already know that $x_{t,i} > 0$ and $x_{t,j} > 0$, so $k \notin \{i, j\}$. Repeating the same argument, for every $\tau < t$, we can show that $i_{\tau} \notin \{i, j\}$, which implies $x_{t+1,j} = x_{t,j} = x_{0,j}$. Using same argument as done for Case 1, we can show that $x_{t+1,i} \ge \min(1, x_{t,j})$ and $\gamma \ge \min(\{a\} \cup \cA)$.

\paragraph{\textbf{Case 3:}} Definition~\ref{def:warmup} condition \eqref{def:warmup:1} violated at $t$, i.e., there is an agent $i$ with $x_{t,i} > \frac{n^2}{4(n-1)}$, but conditions \eqref{def:warmup:2} and condition \eqref{def:warmup:3} are satisfied. After agent $i$ made the move at time $t$, $x_{t+1,i} \le \frac{n^2}{4(n-1)}$ at time $t+1$.

We have $s_{t,-i} > 0$ (by condition \eqref{def:warmup:2}) and $s_{t,-i} \le s_{t+1} \le \frac{n^2}{2(n-1)}$ because $\gamma = \min_{j} s_{t+1,-j} < \frac{1}{n-1}$, and which implies that $s_{t+1} \le \gamma + \max_j x_{t+1,j} \le \frac{1}{n-1} + \frac{n^2}{4 (n-1)} \le \frac{n^2}{2(n-1)}$.

We can lower bound $\gamma$ as a function of $s_{t+1,-i}$. First, notice that $\gamma \ge \min(s_{t+1,-i}, x_{t+1,i}) = \min(s_{t,-i}, x_{t+1,i})$. Further, we can lower bound $x_{t+1,i}$ as $x_{t+1,i} \ge \min(1, \frac{s_{t,-i}}{2})$ because, if $x_{t+1,i} \le 1$, then from the first-order condition we have
\begin{multline*}
    x_{t+1,i} \ge \sqrt{s_{t,-i}} \left( \frac{n}{\sqrt{n-1}} - \sqrt{s_{t,-i}} \right) \ge \sqrt{s_{t,-i}} \left( \frac{n}{\sqrt{n-1}} - \frac{n}{\sqrt{2} \sqrt{n-1}} \right) \\
    \ge \frac{\sqrt{s_{t,-i}}}{2} \ge \min\left(1, \frac{s_{t,-i}}{2} \right).
\end{multline*}
Therefore, let us focus on lower bounding $s_{t,-i}$.

If $t=0$, then we trivially have $s_{t,-i} = \sum_{j \neq i} x_{0,j} \ge \min(\cA)$ (and we know that $s_{t,-i} > 0$, so there is an agent $j \neq i$ with positive $x_{0,j}$).

Let $t > 0$. In the proof of Lemma~\ref{lm:part1}, we argued that if an agent $j$ plays at least one move before time $\tau$, for any $\tau$, then $x_{\tau,j} \le \frac{n^2}{4(n-1)}$. As $x_{t,i} > \frac{n^2}{4(n-1)}$, so $i_{\tau} \neq i$ for all $\tau < t$, which implies $x_{\tau,i} = x_{0,i}$ for all $\tau \le t$. 

Let $j = i_{t-1} \neq i$ be the agent who made the move at time $t-1$. Let $\beta = \sum_{k \neq i,j} x_{t-1, k}$. So, $s_{t-1, -j} = \beta + x_{t-1,i} = \beta + x_{0,i}$ and $s_{t, -i} = \beta + x_{t, j}$. Let's do a case analysis on the value of $\beta$ and $x_{t, j}$
\begin{itemize}
    \item Let $x_{t,j} \ge 1$ or $\beta \ge 1$. Then $s_{t, -i} = \beta + x_{t,j} \ge 1$.
    \item Let $x_{t,j} < 1$ and $\beta < 1$. The first-order condition for agent $j$ is
    \begin{align*}
        \frac{s_{t-1, -j}}{(x_{t,j} + s_{t-1, -j})^2} = \frac{n-1}{n^2} c'(x_{t,j}).
    \end{align*}
    Notice that $\frac{s_{t-1, -j}}{(x_{t,j} + s_{t-1, -j})^2}$ decreases with $s_{t-1, -j}$ because the derivative w.r.t. $s_{t-1, -j}$ is
    \begin{align*}
        \frac{\partial }{\partial s_{t-1, -j}} \left( \frac{s_{t-1, -j}}{(x_{t,j} + s_{t-1, -j})^2} \right) = \frac{(x_{t,j} + s_{t-1, -j}) - 2 s_{t-1, -j}}{(x_{t,j} + s_{t-1, -j})^3} < 0,
    \end{align*}
    as $x_{t,j} < 1$ and $s_{t-1, -j} \ge x_{0,i} \ge \frac{n^2}{4(n-1)} \ge 1$. This implies that if $\beta$ increases, then $s_{t-1, -j}$ increases, then $x_{t,j}$ decreases.

    If $\kappa = \infty$ (i.e., $c'(0) = 0$), then we get a lower bound $x_{t,j} = BR(x_{0,i} + \beta) \ge BR(x_{0,i} + 1) > 0$. On the other hand, if $\kappa < \infty$, then either $\beta \ge (\kappa - x_{0,i})/2 > 0$ or $x_{t,j} = BR(x_{0,i} + \beta) \ge BR( (\kappa + x_{0,i})/2) > 0$. So, $s_{t, -i} = \beta + x_{t,j} \ge \max(\beta, x_{t,j}) \ge \min((\kappa - x_{0,i})/2, BR( (\kappa + x_{0,i})/2))$.
\end{itemize}
\qed\end{proof}

%% file: proofs/lipschitz.tex
\begin{proof}[Lemma~\ref{lm:lipschitz}]
Fix an arbitrary agent $i$. Let $\hepsilon = 2(1+K)\epsilon$ Let $u_+ = u_i(BR(s_{-i}), s_{-i})$ and $u_- = u_i(x_i, s_{-i})$. We want to prove that
\begin{align*}
    u_- \ge (1 - \hepsilon) u_+ \Longleftrightarrow \frac{u_-}{u_+} \ge 1 - \hepsilon.
\end{align*}

As $|| \bx - \bx^* ||_1 \le \epsilon$, we have $|s_{-i} - s_{-i}^*| \le \epsilon \implies s_{-i} \in [(n-1) - \epsilon, (n-1)+\epsilon]$. Also, as $|| \bx - \bx^* ||_1 \le \epsilon$, we have $x_i \in [1-\epsilon, 1+\epsilon]$, and as $| BR(s_{-i}) - x^*_i | \le \epsilon$, we have $BR(s_{-i}) \in [1-\epsilon, 1+\epsilon]$. We can lower bound $u_- $ as
\begin{align*}
    u_- &= u_i(x_i, s_{-i}) = \frac{x_i}{x_i + s_{-i}} - \frac{n-1}{n^2} c(x_i) \\
    &\ge \min_{ \substack{ y \in [1-\epsilon, 1+\epsilon] \\ z \in [(n-1) - \epsilon, (n-1)+\epsilon] } } \left( \frac{y}{y + z} - \frac{n-1}{n^2} c(y) \right) \\
    &\ge \min_{ \substack{ y \in [1-\epsilon, 1+\epsilon] \\ z \in [(n-1) - \epsilon, (n-1)+\epsilon] } } \frac{y}{y + z} - \max_{ y \in [1-\epsilon, 1+\epsilon] } \frac{n-1}{n^2} c(y)  \\
    &\ge \min_{ y \in [1-\epsilon, 1+\epsilon] } \frac{y}{y + (n-1) + \epsilon} - \frac{n-1}{n^2} c(1+\epsilon) \\
    &\ge \frac{1 - \epsilon}{n} - \frac{n-1}{n^2} (c(1) + K\epsilon) \\
    &\ge \frac{n - (n-1)c(1)}{n^2} - \frac{\epsilon(n + (n-1)K)}{n^2}
\end{align*}
Similarly, we can upper bound $u_+ $ as
\begin{align*}
    u_+ &= u_i(BR(s_{-i}), s_{-i}) = \frac{BR(s_{-i})}{BR(s_{-i}) + s_{-i}} - \frac{n-1}{n^2} c(BR(s_{-i})) \\
    &\le \max_{ \substack{ y \in [1-\epsilon, 1+\epsilon] \\ z \in [(n-1) - \epsilon, (n-1)+\epsilon] } } \left( \frac{y}{y + z} - \frac{n-1}{n^2} c(y) \right) \\
    &\le \max_{ \substack{ y \in [1-\epsilon, 1+\epsilon] \\ z \in [(n-1) - \epsilon, (n-1)+\epsilon] } } \frac{y}{y + z} - \min_{ y \in [1-\epsilon, 1+\epsilon] } \frac{n-1}{n^2} c(y)  \\
    &\le \max_{ y \in [1-\epsilon, 1+\epsilon] } \frac{y}{y + (n-1) - \epsilon} - \frac{n-1}{n^2} c(1-\epsilon) \\
    &\le \frac{1 + \epsilon}{n} - \frac{n-1}{n^2} (c(1) - K\epsilon) \\
    &\le \frac{n - (n-1)c(1)}{n^2} + \frac{\epsilon(n + (n-1)K)}{n^2}
\end{align*}
Putting these two bounds together, we get
\begin{align*}
    \frac{u_-}{u_+} &\ge \frac{n - (n-1)c(1) - \epsilon(n + (n-1)K)}{n - (n-1)c(1) + \epsilon(n + (n-1)K)} = \frac{1 - \epsilon \frac{n + (n-1)K}{n - (n-1)c(1)}}{1 + \epsilon \frac{n + (n-1)K}{n - (n-1)c(1)}} \\
    &\ge \frac{1 - \epsilon (n + (n-1)K)}{1 + \epsilon (n + (n-1)K)}, \\
    &\qquad\qquad\text{ because $c(1) \le 1$ as $c(0) = 0$ and $c'(y) \le c'(1) = 1$ for $y \le 1$,} \\
    &\ge \frac{1 - \epsilon n(1+K)}{1 + \epsilon n(1+K)} = \frac{1 - \hepsilon/2}{1 + \hepsilon/2} \ge (1 - \hepsilon/2)^2, \\
    &\qquad\qquad\text{ because $\hepsilon \le 1$ and $1 \ge 1-y^2 \Longleftrightarrow \frac{1}{1+y} \ge 1-y$ for $0 \le y \le 1$,} \\
    &\ge 1 - \hepsilon, \text{ as required.}
\end{align*}
As $c'(z) \ge c'(1) = 1$ for $z \in [1, 1+\epsilon]$, so $K \ge 1$ and $2(1+K)n\epsilon \le 4Kn\epsilon$.
\qed\end{proof}

%% file: proofs/hom-bestcase.tex
\begin{proof}[Theorem~\ref{thm:hom:bestcase}]
The proof for upper bound is straightforward as we already have done most of the work in the proof of Theorem~\ref{thm:hom}. First, to complete the warm-up phase (Definition~\ref{def:warmup}), we need every agent to play at least once, and additionally, at most one more transition, as argued in the proof of Lemma~\ref{lm:part1}. This can be completed in $n+1$ time steps by selecting agents in a round robin manner. Second, the number of non-redundant transitions (consecutive transitions by the same agent are redundant) required after the warm-up phase to reach the a phase of the BR dynamics that corresponds to the discounted-sum dynamics is bounded above by $O(\log\log(\frac{1}{\gamma}))$ as argued in the proof of Lemma~\ref{lm:good_domain3}. And this can be completed in $O(\log\log(\frac{1}{\gamma}))$ time by again selecting the agents in a round robin manner. Third, using the upper bound for the best-case selection model for the discount-sum dynamics (Lemma~\ref{lm:dissum:bestcase}), we reach close to the equilibrium output profile measured by $\ell_1$-distance in $O(n \log(\frac{n}{\epsilon}))$ time. Finally, using Lemma~\ref{lm:lipschitz} completes our proof.

\sloppy
Let us now prove the lower bound. We will separately show bounds of $\Omega(n)$, $\Omega(\log\log(\frac{1}{\gamma}) - \log\log(n))$, and $\Omega(\frac{1}{\log(n)}\log(\frac{1}{\epsilon}))$, which together imply a lower bound of $\Omega(\max(n, \log\log(\frac{1}{\gamma}) - \log\log(n), \frac{1}{\log(n)}\log(\frac{1}{\epsilon})) = \Omega(n + \log\log(\frac{1}{\gamma}) + \frac{1}{\log(n)}\log(\frac{1}{\epsilon}))$. 
Note that these bounds are in the worst case over the class of all convex cost functions. In particular, except for the lower bound of $\Omega(n)$ that holds for all cost functions, for every $\epsilon$ and $\gamma$, there exists a convex cost function that reaches an approximate equilibrium after just one BR transition by each agent. In particular, if $c'(z) = z^r$ for all $z \ge 0$ and $r \rightarrow \infty$, then $BR(s_{-i}) \rightarrow 1$ for any $i$ and $s_{-i}$. So, we will prove the lower bounds w.r.t. $\gamma$ and $\epsilon$ using the linear cost function $c'(z) = 1$ for all $z \ge 0$. 
We leave tighter analysis of lower and upper bounds for specific classes of convex cost functions for future work.

\paragraph{Proof for $\Omega(n)$.}
Observe that, to reach an approximate equilibrium, every agent must make at least one best-response move. For example, it is easy to check that there can be no approximate equilibrium with any agent producing an output of $n$: if an agent is playing $n$, then everyone else must play $\le 1$ in an equilibrium (see Lemma~\ref{lm:good_domain1} for a formal argument); if everyone else is playing $\le 1$, then this agent would want to deviate and not play $n$. So, if every agent starts with an output of $n$, then each agent must make at least one move to reach an approximate equilibrium. So, we get a lower bound of $\Omega(n)$ as required. Notice that the same argument also implies that for the randomized model, using Lemma~\ref{lm:coupon}, we get the lower bound of $\Omega(\frac{1}{L} \log(n \delta) )$ w.p. $1-\delta$.

\paragraph{Proof for $\Omega(\log\log(\frac{1}{\gamma}) - \log\log(n))$.} Notice that, at the equilibrium, the total output is $n$. So, for an output profile $\bx_t$ to be at an $\epsilon$-approximate equilibrium for small enough $\epsilon$, we need $s_t \ge \frac{n}{2}$ (Lemma~\ref{lm:revlipschitz} formally proves a similar result). Starting from small $s_0 = \gamma$, we lower bound the time it takes to reach $s_t \ge \frac{n}{2}$. At time $t$, for any agent $i$, the first order condition, equation \eqref{eq:br}, for an agent with a linear cost function is
\begin{multline*}
    \frac{s_{t-1,-i}}{s_{t}^2} = \frac{n-1}{n^2} \Longleftrightarrow s_t = \frac{n}{\sqrt{n-1}} \sqrt{s_{t-1}} \\
    \implies s_t \le n \sqrt{s_{t-1}} \le n^{1+\frac{1}{2}} s_{t-2}^{\frac{1}{2^2}} \le n^{1+\frac{1}{2}+\frac{1}{2^2}} s_{t-3}^{\frac{1}{2^4}} \ldots \le n^2 s_0^{\frac{1}{2^t}} \le n^2 \gamma^{\frac{1}{2^t}}.
\end{multline*}
So, $s_t \ge \frac{n}{2} \implies n^2 \gamma^{\frac{1}{2^t}} \ge \frac{n}{2}$, or 
\begin{align*}
    s_t < \frac{n}{2} \impliedby n^2 \gamma^{\frac{1}{2^t}} < \frac{n}{2} \Longleftrightarrow \frac{1}{2^t} \lg(\gamma) < \lg\left( \frac{1}{2n} \right) \Longleftrightarrow t < \lg\lg\left(\frac{1}{\gamma}\right) - \lg\lg(2n).
\end{align*}
So, we need $t \ge \lg\lg\left(\frac{1}{\gamma}\right) - \lg\lg(2n)$ to ensure $s_t \ge \frac{n}{2}$ and that we are close to the equilibrium.

\paragraph{Proof for $\Omega(\frac{1}{\log(n)} \log(\frac{1}{\epsilon}))$.} 
Before we prove the lower bound w.r.t. $\epsilon$, let us prove the following lemma.
\begin{lemma}\label{lm:revlipschitz}
    Given an output profile $\bx = (x_i)_{i \in [n]}$ and the equilibrium profile $\bx^* = (1, \ldots, 1)$, if $|| \bx - \bx^* ||_1 \le \epsilon$ and $|BR(s_{-i}) - 1| \le \epsilon$, then $u_i(x_i, s_{-i}) < \left(1 - \frac{(x_i - BR(s_{-i}))^2}{8 n^2}  \right) u_i(BR(s_{-i}), s_{-i}) $, for $0 \le \epsilon \le 1$ and $x_i \neq BR(s_{-i})$.
\end{lemma}
\begin{proof}
    We prove the lemma using strong concavity of the utility function near the equilibrium point. From equation \eqref{eq:ddu}, we have
    \begin{align*}
        - \frac{\partial^2 u_i(z, s_{-i})}{\partial z^2} &= \frac{2 s_{-i} }{(z + s_{-i})^3} + \frac{n-1}{n^2} c_i''(z) \ge \frac{2 s_{-i} }{(z + s_{-i})^3}.
    \end{align*}
    As $|| \bx - \bx^* ||_1 \le \epsilon$ and $|BR(s_{-i}) - 1| \le \epsilon$, we have $x_i \in [1-\epsilon, 1+\epsilon]$, $BR(s_{-i}) \in [1-\epsilon, 1+\epsilon]$, and $s_{-i} \in [n-1-\epsilon, n-1+\epsilon]$. So, for any $z \in [1-\epsilon, 1+\epsilon]$, we have
    \begin{align*}
        - \frac{\partial^2 u_i(z, s_{-i})}{\partial z^2} &\ge \min_{ \substack{z \in [1-\epsilon, 1+\epsilon] \\ s_i \in [n-1-\epsilon, n-1+\epsilon]} } \frac{2 s_{-i} }{(z + s_{-i})^3} = \frac{2 (n-1 + \epsilon)}{(n + 2 \epsilon)^3} \ge \frac{1}{4n^2},
    \end{align*}
    for $n \ge 2$ and $\epsilon \in [0,1]$. At $z = BR(s_{-i})$, by the first order condition, we have $\left. \frac{\partial u_i(z, s_{-i})}{\partial z} \right|_{z = BR(s_{-i})} = 0$. So, using the strong convexity of $-u_i(z, s_{-i})$ w.r.t. $z$, we can write
    \begin{align*}
        - u_i(x_i, s_{-i}) &\ge -u_i(BR(s_{-i}), s_{-i}) + (x_i - BR(s_{-i})) \left. \frac{\partial u_i(z, s_{-i})}{\partial z} \right|_{z = BR(s_{-i})} \\
        &\qquad \qquad + \frac{1}{8 n^2} (x_i - BR(s_{-i}))^2 \\
        u_i(x_i, s_{-i}) &\le u_i(BR(s_{-i}), s_{-i}) - \frac{1}{8 n^2} (x_i - BR(s_{-i}))^2 \\
            &< \left(1 - \frac{(x_i - BR(s_{-i}))^2}{8 n^2}  \right) u_i(BR(s_{-i}), s_{-i}),
    \end{align*}
    as $u_i(BR(s_{-i}), s_{-i}) < 1$ and $x_i \neq BR(s_{-i})$, which completes the proof.
\qed\end{proof}

Given Lemma~\ref{lm:revlipschitz}, we can prove that an action profile $\bx$ is not an $\epsilon$-approximate equilibrium profile by showing that there is some agent $i$ who has large deviation, i.e., $(x_i - BR(s_{-i}))^2$ is large.

Next, in Lemma~\ref{lm:hom:bestcase:1}, we show that if there are two agents that have output at least $\alpha$ away from the equilibrium output of $1$ at time $t$, then there are at least two agents that are at least $\alpha/(2n)$ away from $1$ at time $t+1$.
\begin{lemma}\label{lm:hom:bestcase:1}
    Assuming $n \ge 3$. At time $t$, if there exists agents $i$ and $j \neq i$ such that $|x_{t,i} - 1| \ge \alpha$ and $|x_{t,j} - 1| \ge \alpha$ for some $\alpha \in (0,1)$, then there exists agents $i'$ and $j' \neq i'$ such that $|x_{t+1,i'} - 1| \ge \frac{\alpha}{2n}$ and $|x_{t+1,j'} - 1| \ge \frac{\alpha}{2n}$.
\end{lemma}
\begin{proof}
    Let $i$ and $j \neq i$ denote the agents that are at least $\alpha$ away from the equilibrium, i.e., $|x_{t,i} - 1| \ge \alpha$ and $|x_{t,j} - 1| \ge \alpha$ at time $t$. If the agent making the BR move at time $t$ is not in $\{i,j\}$, then we trivially have two agents---agents $i$ and $j$---at time $t+1$ who are at least $\alpha \ge \frac{\alpha}{2n}$ away from the equilibrium. With a similar argument, if there are three or more agents that are at least $\frac{\alpha}{2n}$ away from the equilibrium at time $t$, we will have the required condition at time $t+1$. So, let us assume w.l.o.g. that agent $i$ makes the transition at time $t$ and all agents $k \notin \{i,j\}$ have $|x_{t,k} - 1| < \frac{\alpha}{2n}$.

    As $|x_{t,k} - 1| < \frac{\alpha}{2n}$ for all $k \notin \{i,j\}$, we have $| \sum_{k \notin \{i,j\}} x_{t,k} - (n-2) | < \frac{(n-2) \alpha}{2n} < \frac{\alpha}{2}$. We also know that $|x_{t,j} - 1| \ge \alpha$. Putting together, we have $|s_{t,-i} - (n-1)| > \frac{\alpha}{2}$, or $s_{t,-i} = (n-1) \pm \beta$ for $\beta > \frac{\alpha}{2}$. Let us do a case analysis based on the sign of $s_{t,-i} - (n-1)$.
    \begin{itemize}
        \item $s_{t,-i} - (n-1) < 0 \Longleftrightarrow s_{t,-i} = (n-1) - \beta$. Using the first order condition for agent $i$ (with a linear cost function $c'(z) = 1$ for $z \ge 0$), we can write $x_{t+1,i}$ as
        \begin{align*}
            x_{t+1,i} &= n \sqrt{\frac{s_{t,-i}}{n-1}} - s_{t,-i} = n \sqrt{\frac{(n-1) - \beta}{n-1}} - ((n-1) - \beta) \\
            &= n \sqrt{1 - \frac{\beta}{n-1}} - (n-1) + \beta \\
            &\ge n \left(1 - \frac{\beta}{n-1}\right) - (n-1) + \beta, \text{ because $0 < 1 - \frac{\beta}{n-1} < 1$,} \\
            &\ge 1 + \frac{n-2}{n-1} \beta \ge 1 + \frac{\beta}{2} \ge 1 + \frac{\alpha}{4} \ge 1 + \frac{\alpha}{2n}.
        \end{align*}

        \item $s_{t,-i} - (n-1) > 0 \Longleftrightarrow s_{t,-i} = (n-1) + \beta$. Again using the first order condition for agent $i$, we can write $x_{t+1,i}$ as
        \begin{align*}
            x_{t+1,i} &= n \sqrt{\frac{s_{t,-i}}{n-1}} - s_{t,-i} = n \sqrt{\frac{(n-1) + \beta}{n-1}} - ((n-1) + \beta) \\
            &= n \sqrt{1 + \frac{\beta}{n-1}} - (n-1) - \beta \\
            &\le n \left(1 + \frac{\beta}{n-1}\right) - (n-1) - \beta, \text{ because $1 + \frac{\beta}{n-1} > 1$,} \\
            &\le 1 - \frac{n-2}{n-1} \beta \le 1 - \frac{\beta}{2} \le 1 - \frac{\alpha}{4} \le 1 - \frac{\alpha}{2n}.
        \end{align*}
    \end{itemize}
    Putting both cases together, we have $|x_{t+1,i} - 1| \ge \frac{\alpha}{2n}$; we also know that $|x_{t+1,j} - 1| = |x_{t,j} - 1| \ge \alpha \ge \frac{\alpha}{2n}$. So, we have the two agents at time $t+1$ that satisfy the required condition.
\qed\end{proof}
Lemma~\ref{lm:hom:bestcase:1} above implies that if we start from $\bx_0 = (\frac{1}{2}, \frac{1}{2}, 1, \ldots, 1)$ that has two agents who are producing output at least $\frac{1}{2}$ away from equilibrium output of $1$, then after $t$ time steps, $\bx_t$ will have at least two agents $i$ and $j \neq i$ with $|x_{t+1,k} - 1| \ge \frac{1}{2(2n)^t}$ for $k \in \{i,j\}$. Let $\alpha = \frac{1}{2(2n)^t}$ for conciseness. Now, let $z_{t,k} = x_{t,k}-1$ for all $t \ge 0$ and $k \in [n]$. As $|x_{t+1,k} - 1| \ge \alpha$ for $k \in \{i,j\}$, we have $|z_{t,k}| \ge \alpha$ for $k \in \{i,j\}$. Further, from Lemmas~\ref{lm:good_domain1},~\ref{lm:good_domain2},~and~\ref{lm:good_domain3}, we know that $\bz_t = (z_{t,k})_{k \in [n]}$ follows the discounted-sum dynamics presented in Section~\ref{sec:dissum}. Then, as $|z_{t,k}| \ge \alpha$ for $k \in \{i,j\}$, the potential function given in equation~\eqref{eq:potential} satisfies $f(\bz_t) \ge \alpha$. In particular, one of the two sides of the potential function $f(\bz_t)$ has size at least $\alpha$, so the largest element on the larger side, say $|z_{t,k}|$, must have value at least $\frac{\alpha}{n}$. Moreover, if we select this agent $k$ at time $t$, then $|z_{t+1,k}| \le \frac{|z_{t,k}|}{2}$ as argued in the proof of Lemma~\ref{lm:dissum:random}, which implies 
\begin{align*}
    &|z_{t+1,k}| \le \frac{|z_{t,k}|}{2} \\
    &\implies |z_{t+1,k} - z_{t,k}| \ge \frac{|z_{t,k}|}{2} \ge \frac{\alpha}{2n} = \frac{1}{2(2n)^{t+1}} \\
    &\implies | BR(s_{t,-k}) - x_{t,k} | \ge \frac{1}{2(2n)^{t+1}} \\
    &\implies u_k(x_{t,k}, s_{t,-k}) < \left(1 - \frac{1}{2^{t+5} n^{t+3}}  \right) u_i(BR(s_{t,-k}), s_{t,-k}) \text{ by Lemma~\ref{lm:revlipschitz}.}
\end{align*}
Setting $t \le \frac{\log(1/\epsilon)}{\log(2n)} - 5 \implies \frac{1}{2^{t+5} n^{t+3}} \ge \epsilon$ completes the proof.
\qed\end{proof}